\documentclass[10pt]{article}

\usepackage[margin=1in]{geometry}
\usepackage{mathtools}
\usepackage{footnote}
\usepackage{bbm}
\usepackage[noend]{algpseudocode}
\usepackage{tikz}
\usetikzlibrary{positioning,decorations.pathreplacing,backgrounds}
\usepackage{xspace}
\usepackage[ruled,vlined,linesnumbered]{algorithm2e}
\usepackage{makecell}
\usepackage{booktabs}
\usepackage{authblk}
\usepackage{url}
\usepackage{hyperref}
\usepackage{enumerate}
\usepackage{amsmath}
\usepackage{amsthm}
\usepackage{amssymb}
\usepackage{array}
\usepackage{tabularx}
\usepackage{threeparttable}
\usepackage{booktabs}
\usepackage[table]{xcolor}

\newcommand{\treeChildren}{
    \begin{tikzpicture}[
            scale=0.75,
            every node/.style={
            circle,
            draw=none,
            minimum size=2.5em,
            inner sep=0pt,
            font=\normalsize
            },
            every edge/.style={>=Stealth, line width=0.7, draw=black},
            ]
        \node at (0,0.5) {$X$};
        \node at (-0.5,-1) {$\nu_1$};
        \node at (0,-1) {$\hdots$};
        \node at (0.5,-1) {$\nu_k$};
        \draw (-0.5,-0.75) -- (-0.1,0.25);
        \draw (0.5,-0.75) -- (0.1,0.25);
    \end{tikzpicture}
}

\newcommand{\treeChildrenReturn}{
    \begin{tikzpicture}[
            scale=0.75,
            every node/.style={
            circle,
            draw=none,
            minimum size=2.5em,
            inner sep=0pt,
            font=\normalsize
            },
            every edge/.style={>=Stealth, line width=0.7, draw=black},
            ]
        \node at (0,1) {$V'_X(\vect S \setminus \{X\})$};
        \node at (0,0) {$V_X(\vect S)$};
        \node at (-1,-1) {$V_1(\vect S_1)$};
        \node at (0,-1) {$\hdots$};
        \node at (1,-1) {$V_k(\vect S_k)$};
        \draw (-0.75,-0.75) -- (-0.1,-0.25);
        \draw (0.75,-0.75) -- (0.1,-0.25);
        \draw (0,0.25) -- (0, 0.75);
    \end{tikzpicture}
}

\newtheorem{theorem}{Theorem}
\newtheorem{lemma}[theorem]{Lemma}

\newtheorem{proposition}[theorem]{Proposition}

\newtheorem{corollary}[theorem]{Corollary}
\newtheorem{example}[theorem]{Example}
\newtheorem{definition}[theorem]{Definition}

\newtheorem{conjecture}[theorem]{Conjecture}

\definecolor{burgundy}{RGB}{170,30,80}
\newcommand{\eden}[1]{\textcolor{red}{Eden: #1}}
\newcommand{\andrei}[1]{\textcolor{blue}{AD: #1}}
\newcommand{\haozhe}[1]{\textcolor{red}{Haozhe: #1}}
\newcommand{\haozheupdate}[1]{\textcolor{blue}{#1}}

\newcommand{\nop}[1]{}

\newcommand{\N}{\mathbb{N}}
\renewcommand{\H}{\mathcal{H}}
\newcommand{\V}{\mathcal{V}}
\newcommand{\T}{\mathcal{T}}
\newcommand{\E}{\mathcal{E}}

\newcommand{\Z}{\mathbb{Z}}
	
\newcommand{\R}{\mathbb{R}}

\newcommand{\Qtop}{\widehat{Q}}    

\newcommand{\bigO}{\mathcal{O}}

\newcommand{\D}{\mathcal{D}}

\newcommand{\vect}[1]{\mathbf{#1}}

\newcommand{\set}[1]{\left\{#1\right\}}

\newcommand{\defeq}{\stackrel{\text{def}}{=}}

\newcommand{\at}{\mathsf{at}}
\newcommand{\Bound}{\mathsf{Bound}}
\newcommand{\Free}{\mathsf{Free}}
\newcommand{\vars}{\mathsf{vars}}

\newcommand{\roots}{\mathsf{roots}}
\newcommand{\leaves}{\mathsf{leaves}}
\newcommand{\Dom}{\mathsf{Dom}}
\newcommand{\fhtw}{\mathsf{fhtw}}

\newcommand{\pathh}{\mathsf{path}}
\newcommand{\nodes}{\mathsf{nodes}}
\newcommand{\FTD}{\mathsf{FTD}}
\newcommand{\TD}{\mathsf{TD}}
\newcommand{\width}{\mathsf{width}}

\algrenewcommand\algorithmicforall{\textbf{foreach}}
\begin{document}

\title{The Role of Semirings in Incremental View Maintenance}

\author{Eden Chmielewski, Andrei Draghici, Dan Olteanu, Haozhe Zhang}

\affil{University of Zurich}

\date{}

\maketitle

\begin{abstract}
    We study the problem of incremental view maintenance (IVM) under inserts to semiring-annotated databases. The key observation put forward in this paper is that the complexity of the IVM problem depends fundamentally on the underlying semiring.

    We introduce a class of conjunctive queries called p-hierarchical. For a zero-sum free and zero-divisor free commutative semiring $K$, we show that for any p-hierarchical query with fractional hypertree width $\fhtw$ and any insert-only update sequence of length $N$ to an initially empty $K$-database, we can construct a data structure that can be updated in $\bigO(N^{\fhtw-1})$ amortized time and supports the enumeration of the query result with constant delay. In particular, the amortized update time for any p-hierarchical $\alpha$-acyclic query is constant.
    
    For a class of semirings used to model a wide range of computational problems, we give conditional lower bounds showing that any conjunctive query without self-joins that is not p-hierarchical cannot be maintained with constant amortized update time and constant enumeration delay under inserts. This class includes the natural semiring and its generalizations to the provenance and covariance semirings, as well as idempotent and strictly ordered semirings such as the tropical semiring.

    When put together, our upper and lower bounds imply a dichotomy for the insert-only maintenance of conjunctive queries without self-joins over every semiring in our class: such a query can be maintained with constant amortized update time and constant enumeration delay if and only if it is p-hierarchical $\alpha$-acyclic.
    
    Our dichotomy is sandwiched between two known maintenance dichotomies, both for the Boolean semiring. On the one hand, tractable maintenance under inserts-only can only be achieved for the free-connex $\alpha$-acyclic queries, which form a strict superset of the p-hierarchical $\alpha$-acyclic queries. On the other hand, tractable maintenance under both inserts and deletes can only be achieved for the q-hierarchical queries, which form a strict subset of the p-hierarchical $\alpha$-acyclic queries.
\end{abstract}

\nop{
\section*{Acknowledgements}
This work was partially supported by Swiss NSF 200021-231956.
}

\section{Introduction}

The Incremental View Maintenance (IVM) problem for conjunctive queries is fundamental to databases~\cite{IVM:GemsPODS:2024}: 
Given a query and a database subject to a sequence of single-tuple updates, the IVM problem is to maintain the result of the query under the updates and to allow for the enumeration of the tuples in the query result after each update. A key observation that motivates the widespread adoption of IVM approaches is that maintaining the result of a query under updates can be more efficient than re-evaluating the query from scratch.

Several IVM systems from academia and industry, including DBToaster~\cite{DBT:VLDBJ:2014}, DynYannaka\-kis~\cite{DynYannakakis:SIGMOD:2017}, F-IVM~\cite{FIVM,IVM:GemsPODS:2024}, DBSP~\cite{DBSP:VLDB:2023},  CROWN~\cite{CROWN}, and RAIVM~\cite{RAIVM2026}, maintain queries under updates~\cite{IVM:GemsPODS:2024}. The deployment of such systems is rarely restricted to maintaining plain conjunctive queries over standard relational databases. Instead, IVM systems often maintain the provenance, multiplicities, or covariance information associated with the query result. Such sophisticated tasks can be modeled as query maintenance over 
$K$-databases~\cite{green2007provenance}, where tuples are mapped to payloads from a commutative semiring $K$.
This calls for a systematic study of IVM over $K$-databases. We say that a query is \emph{tractable} if it can be maintained with constant amortized update time and constant enumeration delay; Sec.~\ref{sec:prelims} defines these measures. The main question addressed in this paper is whether the choice of the semiring $K$ can influence the tractability of query maintenance under inserts: {\em Are there conjunctive queries that are tractable over some semirings and not over others?}

\begin{figure}[t]
    \centering
    \begin{tikzpicture}[scale=0.82, every node/.style={transform shape},
    class label/.style={anchor=west, font=\normalsize\bfseries},
    example/.style={anchor=east, font=\small},
    dichotomy label/.style={anchor=west, font=\small, align=left, inner xsep=4pt},
    leader/.style={gray!60, line width=0.5pt},
]
    \draw[rounded corners=7pt, densely dotted, line width=0.8pt, draw=blue!45!black, fill=blue!3]
        (-6.0,-1.0) rectangle (6.0,2.75);
    \node[class label, text=blue!45!black] at (-5.7,2.4) {$\alpha$-acyclic};
    \node[example, text=blue!45!black] at (5.7,2.4) {$Q(X,Z)=\sum_Y R(X,Y)\cdot S(Y,Z)$};

    \draw[rounded corners=7pt, line width=0.8pt, draw=green!45!black, fill=green!4]
        (-5.5,-0.78) rectangle (5.5,1.95);
    \node[class label, text=green!35!black] at (-5.2,1.6) {free-connex $\alpha$-acyclic};
    \node[example, text=green!35!black] at (5.2,1.6) {$Q(X)=\sum_Y S(X,Y)\cdot T(Y)$};

    \draw[rounded corners=7pt, line width=1.2pt, draw=orange!65!black, fill=yellow!22]
        (-5.0,-0.55) rectangle (5.0,1.15);
    \node[class label, text=orange!45!black] at (-4.7,0.8) {p-hierarchical $\alpha$-acyclic};
    \node[example, text=orange!45!black] at (4.7,0.8) {$Q(X,Y)=R(X)\cdot S(X,Y)\cdot T(Y)$};

    \draw[rounded corners=7pt, line width=0.8pt, draw=gray!65!black, fill=white]
        (-4.5,-0.32) rectangle (4.5,0.32);
    \node[class label, text=gray!65!black] at (-4.2,0.0) {q-hierarchical};
    \node[example, text=gray!65!black] at (4.2,0.0) {$Q(X,Y)=R(X,Y)\cdot S(X)$};

    \draw[leader] (5.5,1.6) -- (6.1,1.6) -- (6.4,1.95);
    \node[dichotomy label, text=green!35!black] at (6.4,1.95)
        {insert-only\\Boolean semiring~\cite{CROWN,Free-connex:CSL:2007}};
    \draw[leader] (5.0,0.8) -- (6.4,0.8);
    \node[dichotomy label, text=orange!45!black] at (6.4,0.8)
        {insert-only\\semirings in $\mathcal{K}$ (Corollary~\ref{cor:dichotomy})};
    \draw[leader] (4.5,0.0) -- (6.1,0.0) -- (6.4,-0.35);
    \node[dichotomy label, text=gray!65!black] at (6.4,-0.35)
        {insert-delete\\Boolean semiring~\cite{QHierarchical}};
\end{tikzpicture}
    \caption{Three classes of tractable queries under different maintenance settings, drawn inside the class of $\alpha$-acyclic queries. For each class, we depict an example query that is not in the next smaller class. The labels at the right name the update setting and the semirings of each dichotomy. In the insert-delete setting, q-hierarchical queries remain tractable over a large class of semirings~\cite{insert-delete-all-semirings}.}
    \label{fig:query_hierarchy}
    \vspace*{-1.5em}
\end{figure}

The answer is yes. Consider the query $Q(X)=\sum_Y S(X,Y)\cdot T(Y)$ over a database under insert-only updates. Over the Boolean semiring, $Q(x)$ is true exactly when there is a value $y$ such that both $S(x,y)$ and $T(y)$ are true; under inserts, $Q(x)$ can therefore change only once, from false to true. The query is free-connex  $\alpha$-acyclic and tractable under inserts~\cite{CROWN}. Over the natural semiring, $Q(x)$ instead counts such witnesses $y$, so one insert into $T$ can change the counts of many output $x$-values, and later inserts can change the same counts repeatedly. Our lower bound shows that $Q$ is not tractable under inserts and over the natural semiring unless the Online Matrix-Vector Multiplication (OMv) conjecture~\cite{OMV:STOC:2015} fails. The same query is thus tractable over the Boolean semiring but not over the natural semiring.

A key result of this work, Cor.~\ref{cor:dichotomy}, characterizes the queries that admit tractable maintenance under inserts over a class $\mathcal{K}$ of semirings (defined in Sec.~\ref{sec:class-k}) that includes the natural and the tropical semirings. The characterization is a dichotomy: a conjunctive query without self-joins admits tractable maintenance over any semiring in $\mathcal{K}$ if and only if it belongs to the new class of \emph{p-hierarchical} $\alpha$-acyclic queries, unless the OMv conjecture fails. Table~\ref{tab:semiring-table} lists representative members of $\mathcal{K}$ together with their applications.

Our dichotomy differs from two well-known dichotomies in: (1) the class of tractable queries, (2) the semirings considered, or (3) the update setting (insert-only versus insert-delete). In the insert-only setting, a conjunctive query without self-joins admits tractable maintenance over the Boolean semiring if and only if it is free-connex $\alpha$-acyclic, unless the Boolean matrix multiplication conjecture fails: CROWN~\cite{CROWN} gives the upper bound, the lower bound immediately follows from the super-linear lower bound on the preprocessing time for the static evaluation of queries that are not free-connex $\alpha$-acyclic~\cite{Free-connex:CSL:2007} (the reduction is given in the proof of Prop.~\ref{prop:boolean-insert-only}). In the insert-delete setting, a conjunctive query without self-joins admits tractable maintenance over the Boolean semiring if and only if it is q-hierarchical, unless the OMv conjecture fails~\cite{QHierarchical}.

\nop{
Second, our dichotomy does not apply to the Boolean semiring.
\haozheupdate{For $\mathbb{B}$-databases, CROWN~\cite{CROWN} shows that every free-connex $\alpha$-acyclic query admits tractable maintenance under inserts. For the converse direction, consider a non-free-connex $\alpha$-acyclic query $Q$ without self-joins, and suppose that $Q$ can be maintained with constant amortized update time and constant enumeration delay. Given a static database $D$ of size $N$, we can construct the data structure for $Q(D)$ by inserting the $N$ tuples of $D$ into an initially empty database. This takes $\bigO(N)$ total time, after which the tuples in $Q(D)$ can be enumerated with constant delay. This contradicts the static lower bound of Bagan et al.~\cite{Free-connex:CSL:2007}, unless the Boolean Matrix Multiplication conjecture fails. Together, these results give a dichotomy for insert-only maintenance over the Boolean semiring, conditioned on the Boolean Matrix Multiplication conjecture: An $\alpha$-acyclic query without self-joins admits tractable maintenance if and only if it is free-connex.}
}

Fig.~\ref{fig:query_hierarchy} shows the tractable query classes of the three dichotomies and, for each class, an example query outside the next smaller class. The classes are strictly nested: q-hierarchical $\subsetneq$ p-hierarchical $\alpha$-acyclic $\subsetneq$ free-connex $\alpha$-acyclic. The dichotomies themselves are not nested, since they concern different update settings and semirings. In particular, in the insert-only setting, strictly fewer queries admit tractable maintenance over the semirings in $\mathcal{K}$ than over the Boolean semiring, unless the OMv conjecture fails.
\nop{
All $\alpha$-acyclic p-hierarchical queries are free-connex, yet there are free-connex $\alpha$-acyclic queries that are not p-hierarchical. On the other hand, all q-hierarchical queries admit tractable maintenance under inserts and deletes to $\mathbb{B}$-databases, while all non-q-hierarchical queries without self-joins  cannot admit tractable maintenance~\cite{QHierarchical}. Furthermore, all q-hierarchical queries are $\alpha$-acyclic and p-hierarchical, while there are acyclic p-hierarchical queries that are not q-hierarchical. Our dichotomy for $\alpha$-acyclic p-hierarchical queries is thus strictly sandwiched between the dichotomies for free-connex $\alpha$-acyclic queries and q-hierarchical queries, as illustrated in Fig.~\ref{fig:query_hierarchy}.
}

\begin{table}[!t]
    \centering
    \footnotesize
    \begin{threeparttable}
    \makeatletter
    \newcommand{\semiringrowrule}{\noalign{\gdef\CT@arc@{\color{black!35}}}\cmidrule[0.15pt](lr){1-3}\noalign{\global\let\CT@arc@\relax}}
    \makeatother
    \setlength{\tabcolsep}{4pt}
    \renewcommand{\arraystretch}{1.05}
    \begin{tabularx}{\textwidth}{@{}>{\raggedright\arraybackslash}p{0.27\textwidth}>{\raggedright\arraybackslash}p{0.24\textwidth}>{\raggedright\arraybackslash}X@{}}
        \toprule
        \textbf{Semiring} & \textbf{\mbox{Why in $\mathcal{K}$ (Def.~\ref{def:calK})}} & \textbf{Applications} \\
        \midrule
        Boolean $(\mathbb{B}, \vee, \wedge)$ & not in $\mathcal{K}$ & satisfiability and CSPs~\cite{FAQ}, graph reachability~\cite{NetworkProblems} \\
        \semiringrowrule
        Natural $(\mathbb{N}, +, \cdot)$ & $\mathbb{N}$ & tuple multiplicities~\cite{green2007provenance}, \#SAT~\cite{FAQ}, network analysis~\cite{FAQ} \\
        \semiringrowrule
        Real sum-product \ $(\mathbb{R}_{\geq 0}, +, \cdot)$ & hom.\ from $\mathbb{N}$ & marginal inference in PGMs~\cite{FAQ}, permanent computation~\cite{FAQ} \\
        \semiringrowrule
        Min-product $(\mathbb{R}_{> 0} \cup \{+\infty\}, \min, \cdot)$ & ISO & minimum product paths \\
        \semiringrowrule
        Max-product (Viterbi) $(\mathbb{R}_{\geq 0}, \max, \cdot)$ & ISO & maximum likelihood decoding~\cite{FAQ}, MAP estimation in PGMs~\cite{FAQ}, maximum reliability~\cite{NetworkProblems} \\
        \semiringrowrule
        Tropical \  $(\mathbb{R} \cup \{+\infty\}, \min, +)$ & ISO & single-source and all-pairs shortest paths~\cite{ShortestDistance:2002, NetworkProblems} \\
        \semiringrowrule
        Max-plus $(\mathbb{R} \cup \{-\infty\}, \max, +)$ & ISO & single-source and all-pairs longest paths~\cite{NetworkProblems} \\
        \semiringrowrule
        Max-min $(\mathbb{R} \cup \{\pm\infty\}, \max, \min)$ & ISO & widest paths~\cite{NetworkProblems} \\
        \semiringrowrule
        Min-max $(\mathbb{R} \cup \{\pm\infty\}, \min, \max)$ & ISO & minimum weight spanning trees~\cite{MST, NetworkProblems}  \\
        \semiringrowrule
        Provenance $(\mathbb{N}[X], +, \cdot)$ & hom.\ from $\mathbb{N}$ & probabilistic databases~\cite{PDB:Book:2011}, fact attribution~\cite{Banzhaf:SIGMOD:2024, Fact:VLDB:2025, Shapley:PODS:2024, Shapley:SIGREC:2023} \\
        \semiringrowrule
        Covariance\tnote{*} $((\mathbb{N}, \mathbb{N}^m, \mathbb{N}^{m \times m}), +_K, \cdot_K)$ & hom.\ from $\mathbb{N}$ & covariance matrices over databases~\cite{FIVM:SIGMOD:2018, FIVM} \\
        \bottomrule
    \end{tabularx}
    \begin{tablenotes}[flushleft]\scriptsize
    \item[*] For $a = (c_a, \vect s_a, Q_a) \in (\mathbb{N}, \mathbb{N}^m, \mathbb{N}^{m \times m})$ and $b = (c_b, \vect s_b, Q_b) \in (\mathbb{N}, \mathbb{N}^m, \mathbb{N}^{m \times m})$, $a +_K b = (c_a + c_b, \vect s_a + \vect s_b, Q_a + Q_b)$ and $a \cdot_K b = (c_a c_b, c_b \vect s_a + c_a \vect s_b, c_b Q_a + c_a Q_b + \vect s_a \vect s_b^T + \vect s_b \vect s_a^T)$. The covariance semiring over its full support is not zero-divisor free: it has non-zero elements whose multiplication is zero. Yet this corner case does not arise in its previous applications in the database context~\cite{FIVM} (see Appendix~\ref{sec:appendix-prelims}). We therefore safely restrict its use as a semiring that is zero-sum free and zero-divisor free. 
    \end{tablenotes}
    \end{threeparttable}
    \vspace{4pt}
    \caption{Representative semirings in $\mathcal{K}$ and some of their applications. Each semiring is given by its domain, addition, and multiplication; the identities of the two operations are omitted for space. The middle column states which condition of item 2 of Def.~\ref{def:calK} the semiring satisfies: there is an injective homomorphism from the natural semiring into it (Def.~\ref{def:embedding}), or it is idempotent strictly ordered (ISO, Def.~\ref{def:temporal-chain-semiring}). Cor.~\ref{cor:dichotomy} holds for every listed semiring except the Boolean semiring, for which the free-connex $\alpha$-acyclic dichotomy~\cite{CROWN,Free-connex:CSL:2007} holds instead (Prop.~\ref{prop:boolean-insert-only}).}
    \label{tab:semiring-table}
    \vspace*{-2em}
\end{table}
\nop{
The semirings considered in this paper enable practical IVM applications not available to IVM engines that only work with databases over the Boolean semiring. 
For instance, the tropical semiring can be used to express the computation of single-source and all-pairs shortest paths in graphs~\cite{ShortestDistance:2002}. 
The provenance semiring is used to: quantify the contributions of input facts to the query result and explain the query result~\cite{Shapley:SIGREC:2023,Shapley:PODS:2024,Banzhaf:SIGMOD:2024,Fact:VLDB:2025}; capture the lineage over random events in probabilistic databases~\cite{PDB:Book:2011}; and recover bag semantics via tuple multiplicities. The covariance semiring (defined over natural numbers) captures the maintenance of the covariance matrix of features defined by the data columns in the result of queries over databases~\cite{FIVM:SIGMOD:2018,FIVM}.
}

\medskip

{\noindent \bf Contributions.} The contributions of this paper are as follows.

\smallskip

\noindent 1. We introduce the class of p-hierarchical queries (Sec.~\ref{sec:p-hierarchical}), which are exactly those conjunctive queries expressible as full conjunctive queries with body atoms defined by q-hierarchical queries whose free variables occur in all their atoms. We show that the fractional hypertree width of a p-hierarchical query equals that of its bound version. The class of p-hierarchical $\alpha$-acyclic queries is placed strictly between the q-hierarchical and the free-connex $\alpha$-acyclic query classes, as depicted in Fig.~\ref{fig:query_hierarchy}.

\smallskip

\noindent 2. We establish an upper bound for the maintenance of p-hierarchical queries under inserts:
\begin{theorem}
\label{thm:main-upper-bound}
    Consider a conjunctive query $Q$ and a sequence of $N$ single-tuple inserts to an initially empty $K$-database, where $K$ is a zero-sum free and zero-divisor free semiring. If $Q$ is p-hierarchical, then it can be maintained with $\bigO(N^{\fhtw - 1})$ amortized update time and constant enumeration delay, where $\fhtw$ is the fractional hypertree width of the bound version of $Q$.
\end{theorem}
Sec.~\ref{sec:upper-bound} sketches the proof and gives our maintenance algorithm.
Theorem~\ref{thm:main-upper-bound} guarantees that p-hierarchical $\alpha$-acyclic queries admit tractable maintenance, since $\fhtw=1$ for them.

\smallskip

\noindent 3. We prove conditional lower bounds on the amortized update time and the enumeration delay for non-p-hierarchical queries over every semiring in $\mathcal{K}$ (Sec.~\ref{sec:lower-bounds}):

\begin{theorem}\label{thm:main-lower-bound}
    Consider a conjunctive query $Q$ without self-joins and a sequence of $N$ single-tuple inserts to an initially empty $K$-database, where $K \in \mathcal{K}$. If $Q$ is not p-hierarchical, then for any $\gamma > 0$ there is no algorithm that maintains $Q$ with $\bigO(N^{\frac{1}{2}-\gamma})$ amortized update time and $\bigO(N^{\frac{1}{2}-\gamma})$ enumeration delay, unless the OMv conjecture fails.
\end{theorem}


\noindent 4. The second and third contributions imply a dichotomy for the maintenance of conjunctive queries without self-joins over $K$-databases under single-tuple inserts:

\begin{corollary}[Theorems~\ref{thm:main-upper-bound} and~\ref{thm:main-lower-bound}]\label{cor:dichotomy}
    Consider a conjunctive query $Q$ without self-joins and a sequence of $N$ single-tuple inserts to an initially empty $K$-database, where $K\in\mathcal{K}$.
    \begin{itemize}
        \item If $Q$ is $\alpha$-acyclic p-hierarchical, then it can be maintained with constant amortized update time and constant enumeration delay. 
        \item Otherwise, there is no algorithm that maintains $Q$ with $\bigO(1)$ amortized update time and $\bigO(1)$ enumeration delay, unless the OMv conjecture or the Hyperclique conjecture fail.
    \end{itemize}
\end{corollary}

\medskip

{\noindent\bf Related work.}
Dynamic query evaluation has been studied under set semantics well beyond the two dichotomies above: trade-offs between update time and enumeration delay for the triangle query~\cite{IVMeps:ICDT:2019,IVMeps:PODS:2020}, queries with free access patterns~\cite{DynamicWidth}, a mix of static and dynamic relations~\cite{dynamicrelation}, and heavy-light partitioning of the input relations~\cite{H-IVM}. On the lower-bound side, recent work analyzes the maintenance cost as a function of the update sequence~\cite{HuWang:PODS:2025} and gives parameterized conditional lower bounds on the update time~\cite{Qichen:PODS:2026}. Beyond sets, F-IVM maintains q-hierarchical queries with constant update time and constant delay under inserts and deletes over the $\mathbb{Z}$-ring and the covariance ring~\cite{FIVM:SIGMOD:2018,FIVM}, and Mu\~noz et al.~\cite{insert-delete-all-semirings} extend this to every sum-maintainable zero-divisor free semiring. Closest to IVM under inserts is the semi-na\"ive evaluation of recursive Datalog queries, which also propagates deltas. Recent work~\cite{Datalogo:JACM:2024} extends recursive Datalog to $K$-databases and asks under which conditions on $K$ the semi-na\"ive evaluation terminates; it does not consider the cost of processing an insert. Unfolding the recursion of reachability a fixed number of times yields path queries over binary relations. These queries are $\alpha$-acyclic but not free-connex, so they are intractable under inserts already over the Boolean semiring. None of these works studies how the choice of the semiring affects the tractability of maintenance.

The foundations of $K$-databases~\cite{green2007provenance} have been extensively explored across various classical database problems. Early works established the decidability and complexity bounds for the containment and equivalence of conjunctive queries under provenance semiring semantics~\cite{Semiring-query-containment}, and generalized the problem of exact query reformulation using views to semiring-annotated relations~\cite{Z-relations}. Applications with richer metadata require the algebraic evaluation of multiple, potentially dependent annotations drawn from different semirings~\cite{Combining-annotations}. Consistent query answering has been studied over naturally ordered positive semirings~\cite{ConsistentAnswers}, a subclass of the zero-sum-free and zero-divisor-free semirings considered in our work. A separate line of work~\cite{Hier-UnifyingAlgo} has identified hierarchical queries as the tractability boundary for the bag-set maximization problem and given a unifying algorithm over 2-monoids for probabilistic query evaluation, Shapley value computation, and the bag-set maximization problem.

\section{Preliminaries}
\label{sec:prelims}

We introduce concepts and notation used in the rest of the paper, following prior work~\cite{FIVM,H-IVM,InsertsVSDeletes}.

{\noindent \bf Semirings.} A \emph{commutative semiring} is a tuple $(K,+,\cdot,0,1)$, where $K$ is a set, the operations $(+)$ and $(\cdot)$ are binary, $(K,+,0)$ and $(K,\cdot,1)$ are commutative monoids, multiplication distributes over addition, and $0 \cdot a = 0$ for all $a \in K$. We refer to a semiring by its set $K$ when the operations and identities are clear from context. A semiring is \emph{zero-sum free} if $a + b = 0$ implies $a = b = 0$ and \emph{zero-divisor free} if $a \cdot b = 0$ implies $a = 0$ or $b = 0$. For example, the natural semiring $(\N, +, \cdot, 0, 1)$ and the Boolean semiring $(\{\mathsf{true}, \mathsf{false}\}, \lor, \land, \mathsf{false}, \mathsf{true})$ are zero-sum free and zero-divisor free; the integer ring $(\Z,+,\cdot,0,1)$ is not zero-sum free. We assume that the elements of a semiring have constant size and that its operations take constant time. All semirings in this paper are commutative.

\nop{A \emph{commutative ring} $(K,+,\cdot,0,1)$ consists of a set $K$ equipped with two binary operations $(+)$ and $(\cdot)$, such that $(K,+,0)$ is an abelian group, $(K,\cdot,1)$ is a commutative monoid, and multiplication distributes over addition. A \emph{semiring} satisfies the same axioms except that $(K,+,0)$ is required only to be a commutative monoid (i.e., additive inverses need not exist). \eden{In this paper, we only consider commutative semirings that have no additive inverses (zero-sum free). Furthermore, we require that for all the semirings considered, the multiplication of nonzero elements gives a nonzero result (zero-divisor free).} Examples of such semirings are the natural sum-product semiring $(\N, +, \cdot, 0, 1)$ and the Boolean semiring $(\{\mathsf{true}, \mathsf{false}\}, \lor, \land, \mathsf{false}, \mathsf{true})$. We assume without loss of generality that the semirings have elements of constant size and that their operations $(+)$ and $(\cdot)$ take constant time. }

\medskip

{\noindent \bf Data and Queries.} A {\em schema} $\vect X$ is a tuple of variables $(X_1,\ldots,X_n)$, which we also treat as a set when convenient. For each variable $X_i \in \vect X$, let $\Dom(X_i)$ denote its domain. A tuple $\vect x$ of values over the schema $\vect X$ is an element of the set $\Dom(\vect X)=\Dom(X_1)\times\dots\times\Dom(X_n)$. 

Let $(K, +, \cdot, 0, 1)$ be a semiring. A $K$-relation $R$ over a schema $\vect X$ and the semiring $K$ is a function $R : \Dom(\vect X) \to K$ that maps tuples of values over $\vect X$ to elements of $K$~\cite{green2007provenance}. Standard relations are $\mathbb{B}$-relations, relations with multiplicities are $\N$-relations, sparse tensors are $\R$-relations. A tuple $\vect x \in \Dom(\vect X)$ is a \emph{key} and $R(\vect x)$ its \emph{payload} in $R$. When applying set operations to $R$, we treat it as the set of tuples $\vect x$ with $R(\vect x) \neq 0$, which we assume to be finite. The size of $R$, denoted by $|R|$, is the size of this set. A $K$-database $\mathcal{D}$ is a collection of relations over the same semiring $K$. Its size $|\mathcal{D}|$ is the sum of the sizes of its $K$-relations. 

To express conjunctive queries over $K$-databases, the operations $\land$, $\lor$, and $\exists$ are replaced by the semiring operations $\cdot$, $+$, and respectively $+$. This naturally leads to the language of functional aggregate queries over the semiring $K$~\cite{FAQ}:
$Q(\vect{F}) \;=\; \sum_{\vect B} R_1({\vect{X_1}})\cdot\ldots\cdot R_k(\vect{X_k})$,
where: $\sum$ and $(\cdot)$ are the summation and multiplication operations from the semiring $K$; $R_1,\dots,R_k$ are {\em relation symbols}; each $\vect{X_i}$ is a schema; and each $R_i(\vect{X_i})$ is an \emph{atom} of $Q$. The relation symbols of a query are distinct. If several of them refer to the same database relation, that is, $Q$ has self-joins, then each atom gets its own copy of that relation; this changes none of the data complexity upper bounds stated in the paper. The set of variables of $Q$ is  $\vars(Q) \defeq \bigcup_{i\in[k]}\vect{X_i}$. The {\em free} variables of $Q$ are $\vect F = \Free(Q) \subseteq \vars(Q)$, while $\vect B = \Bound(Q) = \vars(Q)\setminus\vect F$ are the {\em bound variables}. If $\vect F = \vars(Q)$, then $Q$ is a full query, and if $\vect F = \emptyset$, then $Q$ is a bound query. By $\at(Q)$ and $\at(Y)$ we denote the set of all atoms of $Q$ and the set of all atoms $R_i(\vect{X_i})$ with $Y\in\vect{X_i}$, respectively. For compactness, we write a set of variables as the concatenation of their names, e.g.,~$\{X,Y,Z\}$ becomes $XYZ$ while $\{X\}$ becomes $X$. The \emph{bound version} of $Q$ is the bound query $\sum_{\vars(Q)} R_1({\vect{X_1}})\cdot\ldots\cdot R_k(\vect{X_k})$. 

The marginalization of variables $\vect Y\subseteq \vect X$ from a relation $R$ with schema $\vect X$, denoted by $S(\vect Z) = \sum_{\vect Y} R(\vect X)$ for $\vect Z = \vect X \setminus \vect Y$, is defined by:
$\forall \vect z\in\Dom(\vect Z): 
S(\vect z) \defeq \sum \{R(\vect x) \mid \vect x\in\Dom(\vect X)\wedge \vect z = \vect x.\vect Z \}$, where $\vect x.\vect Z$ is the restriction of the tuple $\vect x$ to the values of the variables in schema $\vect Z$. The union of two relations $R$ and $S$ with the same schema $\vect X$, denoted by $T = R\cup S$, is defined as: $\forall \vect x\in\Dom(\vect X): T(\vect x) = R(\vect x) + S(\vect x)$. Thus, the semantics of $Q$ are defined as the relation: $\forall \vect f\in\Dom(\vect F): Q(\vect f) = \sum_{\vect x\in\Dom(\vars(Q)): \vect x.\vect F=\vect f} R_1(\vect x.\vect X_1)\cdot\ldots\cdot R_k(\vect x.\vect X_k)$.


\medskip

{\noindent \bf Data Updates.}
We maintain the result of a fixed query $Q$ over a $K$-database $\D$ that is initially empty and receives a sequence of $N$ single-tuple inserts; neither $N$ nor the inserts are known in advance. A single-tuple insert adds a tuple $\vect x$ with payload $k \in K$ to a relation $R$ with schema $\vect X$. We represent it by the delta relation $\delta R : \Dom(\vect X) \to K$ that maps $\vect x$ to $k$ and all other tuples to $0$. We apply it as $R \defeq R \cup \delta R$, which adds $k$ to the payload of $\vect x$ in $R$. After each insert, the tuples of $Q(\D)$ can be enumerated with their payloads.

\medskip

{\noindent \bf Complexity Measures.}
The \emph{update time} is the time to process one single-tuple insert. The \emph{amortized update time} for a sequence of $N$ single-tuple inserts is $f(N)$ if the total time to process all $N$ inserts is upper bounded by $N \cdot f(N)$. The \emph{enumeration delay} is the maximum of: the time from the start of the enumeration and before the first tuple is output, the time between two consecutive output tuples, and the time from the last output tuple and until the end of the enumeration~\cite{Free-connex:CSL:2007}; it is \emph{constant} if it does not depend on $N$. The maintenance of a query is \emph{tractable} if both its amortized update time and its enumeration delay are constant.

\medskip

{\noindent \bf Hypergraphs, Tree Decompositions, and Widths.} We use the standard notions of hypertree decompositions of query hypergraphs~\cite{TreeDecompositions}, the AGM bound~\cite{AGM:2008}, and the fractional hypertree width~\cite{FHTW} of a CQ. For convenience, these have been restated in the appendix. 

\nop{\haozhe{Perhaps define $\alpha$-acyclic and free-connex briefly here.}}

\section{The Class $\mathcal{K}$ of Semirings}
\label{sec:class-k}

In this paper, we consider a class $\mathcal{K}$ of semirings with a wide range of practical applications, as exemplified in Table~\ref{tab:semiring-table}.  
Each semiring in $\mathcal{K}$ has one of two properties: it is idempotent and strictly ordered, or there is an injective homomorphism from the natural semiring to it. A semiring is \emph{idempotent} if $a + a = a$ for all $a \in K$. An idempotent semiring is partially ordered by its \emph{natural order}~\cite{green2007provenance,graphs-dioids}: $a \preceq_K b$ if $a + b = b$, and $a \prec_K b$ if $a \preceq_K b$ and $a \neq b$.

\begin{definition}[Idempotent Strictly Ordered (ISO) Semiring]
\label{def:temporal-chain-semiring}
An idempotent semiring $K$ is \emph{strictly ordered} if there is an infinite chain $0 = c_0 \prec_K c_1 \prec_K c_2 \prec_K \cdots$ in $K$ such that $c_i \cdot c_{i+1} \prec_K c_{i+1} \cdot c_{i+1}$ for all $i \geq 0$.
\end{definition}

\begin{example}\label{ex:iso}
The tropical semiring $(\R \cup \{\infty\}, \min, +, \infty, 0)$ is idempotent, and its natural order is the reverse of $\leq$, since $\min(a,b) = b$ whenever $a \geq b$. It is strictly ordered: take $c_0 = \infty$ and $c_i = -i$ for $i \geq 1$. This chain is increasing in the natural order, $c_0 \cdot c_1 = \infty \prec_K -2 = c_1 \cdot c_1$, and $c_i \cdot c_{i+1} = -(2i+1) \prec_K -(2i+2) = c_{i+1} \cdot c_{i+1}$ for $i \geq 1$. The min-product, max-product, max-plus, max-min, and min-max semirings in Table~\ref{tab:semiring-table} are ISO as well. For min-max, the same reasoning holds as for the tropical semiring. For min-product, its increasing chain is $c_0 = \infty$, and $c_i = 2^{-i}$ for all $i \geq 1$. If the addition operation is $\max$, the semiring is idempotent and its natural order is $\leq$. Its increasing chain is $c_0 = 0_K$, and $c_i = i$ for $i \geq 1$. The Boolean semiring is idempotent with natural order $\mathsf{false} \preceq_K \mathsf{true}$, but it has no infinite chain, so it is not strictly ordered; the same holds for every finite semiring. 
\end{example}

\begin{definition}[Homomorphism from the Natural Semiring]
\label{def:embedding}
A \emph{homomorphism} from the natural semiring $(\N, +, \cdot, 0, 1)$ to a semiring $(K, +_K, \cdot_K, 0_K, 1_K)$ is a map $h : \N \to K$ with $h(0) = 0_K$, $h(1) = 1_K$, $h(a + b) = h(a) +_K h(b)$, and $h(a \cdot b) = h(a) \cdot_K h(b)$ for all $a, b \in \N$.
\end{definition}

Intuitively, an injective homomorphism makes the natural semiring part of $K$: each natural number $n$ corresponds to a distinct element $h(n)$ of $K$, and computing in $\N$ is the same as computing with these elements in $K$.

\begin{example}\label{ex:embedding}
The elements of the real sum-product semiring $(\R_{\geq 0}, +, \cdot, 0, 1)$ are the nonnegative real numbers, with the usual addition and multiplication. The identity map $h(n) = n$ is an injective homomorphism from the natural semiring: every natural number is a nonnegative real number, and the sum and the product of two natural numbers are the same in both semirings.
The elements of the covariance semiring are triples $(c, \vect s, \vect Q)$ of a count $c \in \N$, a vector $\vect s \in \N^m$, and a matrix $\vect Q \in \N^{m \times m}$. The natural semiring is part of it through the count: $h(n) = (n, \vect 0, \vect 0)$ is the count $n$ with a zero vector and a zero matrix.
\end{example}

We are now ready to define the class $\mathcal{K}$ of semirings under consideration in our paper.

\begin{definition}[The Semiring Class $\mathcal{K}$]
\label{def:calK}
A semiring $K$ is in the class $\mathcal{K}$ if
\begin{enumerate}
    \item $K$ is zero-sum free and zero-divisor free, and
    \item $K$ is idempotent strictly ordered, where an element in the chain can be computed in constant time from its previous element in the chain, or there is an injective homomorphism from the natural semiring to $K$.
\end{enumerate}
\end{definition}

The two alternatives of item 2 exclude each other. Suppose that a semiring $K$ satisfies both: it is idempotent, and there is an injective homomorphism $h$ from the natural semiring to $K$. Since $h$ preserves sums and maps $1$ to $1_K$, we have $h(2) = h(1 + 1) = h(1) +_K h(1) = 1_K +_K 1_K$. Since $K$ is idempotent, $1_K +_K 1_K = 1_K = h(1)$. Hence $h(2) = h(1)$, which contradicts the injectivity of $h$. Every semiring in $\mathcal{K}$ thus satisfies exactly one of the two alternatives. We give separate reductions for their lower bounds in Sec.~\ref{sec:lower-bounds}.

Zero-sum freeness means that the semiring does not support additive inverses. This is important for the enumeration of the query result: Our maintenance approach propagates non-zero payloads from input tuples to output tuples, so it can enumerate with constant delay tuples whose annotations are guaranteed non-zero. Theorem~\ref{thm:main-upper-bound} holds for every zero-sum free and zero-divisor free semiring. 
The second condition is key for our lower bound in Sec.~\ref{sec:lower-bounds}, which has to express deletion without additive inverses: a semiring with an injective homomorphism from the natural semiring recovers the effect of a deletion as the difference of two consecutive answers, and an idempotent strictly ordered semiring lets a fresh payload from the chain overwrite a stale one. Theorem~\ref{thm:main-lower-bound} and Cor.~\ref{cor:dichotomy} hold for every semiring in $\mathcal{K}$.
Every semiring in Table~\ref{tab:semiring-table} except the Boolean one belongs to $\mathcal{K}$, and the second column of the table states which condition of item 2 it satisfies; the Boolean semiring satisfies neither.

\nop{
A semiring is \emph{zero-sum free} if $a + b = 0$ implies $a = b = 0$ and \emph{zero-divisor free} if $a \cdot b = 0$ implies $a = 0$ or $b = 0$.
}

\section{P-Hierarchical Queries}
\label{sec:p-hierarchical}

In this section, we introduce the new class of p-hierarchical queries, for which we later give lower and upper bounds on their maintenance. We start by introducing the prior, related class  of q-hierarchical queries, which can be maintained with constant update time and constant enumeration delay under inserts-deletes and over the Boolean semiring~\cite{QHierarchical}. 

\begin{definition}[Q-Hierarchical](\cite[Def.~3.1]{QHierarchical})
A conjunctive query $Q$ is \emph{q-hierarchical} if for any two variables $X,Y\in\vars(Q)$ the following is satisfied:
\begin{enumerate}
    \item\ [hierarchical property] \hspace*{1em}$\at(X) \cap \at(Y) = \emptyset$ or $\at(X) \subseteq \at(Y)$ or $\at(Y) \subseteq \at(X)$, and
    \item\ [q property] \hspace*{5em} if $\at(X) \subsetneq \at(Y)$ and $X \in \Free(Q)$, then $Y \in \Free(Q)$.
\end{enumerate}
\end{definition}

A natural question is whether there is a similarly simple syntactic characterization of all tractable conjunctive queries that can be maintained with constant amortized update time and constant enumeration delay under {\em inserts-only} and {\em over a semiring in $\mathcal{K}$}. This is the class of \emph{p-hierarchical} (short for piecewise-hierarchical) $\alpha$-acyclic queries, as we prove in subsequent sections. The key insight behind this class is that under inserts, the free variables no longer need to satisfy the strict hierarchy required by q-hierarchical queries. Instead, the hierarchical constraints apply exclusively to how bound variables interact with each other and with the free variables.

\begin{definition}[P-Hierarchical Query]
\label{def:p-hierarchical}
    A conjunctive query $Q$ is p-hierarchical if the following is satisfied:
    \begin{enumerate}
        \item\hspace*{-0.75em} [bound-bound interaction] $\forall X, Y \in \Bound(Q)$: $\at(X) \subseteq \at(Y)$ or $\at(Y) \subseteq \at(X)$ or $\at(X) \cap \at(Y) = \emptyset$; and
        \item\hspace*{-0.75em} [bound-free interaction] $\forall X \in \Bound(Q)$, $Y \in \Free(Q)$: $\at(X) \subseteq \at(Y)$ or $\at(X) \cap \at(Y)=\emptyset$
    \end{enumerate} 
\end{definition}

The class of p-hierarchical $\alpha$-acyclic queries is strictly sandwiched between the classes of q-hierarchical queries~\cite{QHierarchical} and of free-connex $\alpha$-acyclic queries~\cite{Free-connex:CSL:2007}, as it strictly includes the former and is included in the latter. \nop{\haozheupdate{Over the Boolean semiring, the free-connex property characterizes tractable maintenance under inserts.} \eden{The proof is deferred to the appendix.}}
For variables $X,Y$, we say that $X$ \emph{dominates} $Y$ iff $\at(X)\supseteq\at(Y)$. A query is \emph{free-dominated} if its free variables occur in all of its atoms.

\begin{example}\label{ex:p-hier-query}
    Any full conjunctive query is trivially p-hierarchical, since it has no bound variables. Two simple examples of free-connex $\alpha$-acyclic queries that are not p-hierarchical are: $Q_1() = \sum_{X,Y}R(X)\cdot S(X,Y)\cdot T(Y)$ and $Q_2(X) = \sum_{Y}S(X,Y)\cdot T(Y)$. $Q_1$ violates condition (1), while $Q_2$ violates condition (2).
    A simple p-hierarchical $\alpha$-acyclic query that is not q-hierarchical is the full variant of $Q_1$: $Q'_1(X,Y) = R(X)\cdot S(X,Y)\cdot T(Y)$. This is because $\at(X)$ and $\at(Y)$ overlap but one is not included in the other. Note that $Q_1$ and $Q_2$ are also not q-hierarchical: $Q_1$ is not even hierarchical, while $Q_2$ is hierarchical but does not have the $q$ property.

    Throughout this paper, we use as a running example the p-hierarchical $\alpha$-acyclic query
    \[Q(A,C) = \sum_{B,D,E,F,G,H}R_1(A,B,D,E) \cdot R_2(A,B,D,F) \cdot R_3(A,B,G) \cdot R_4(C) \cdot R_5(A,C,H).\]
    Notice that $B$ dominates $D$ and $G$, $D$ dominates $E$ and $F$, and $H$ has an atom set that is disjoint from the atom sets of all other bound variables. This satisfies condition (1). Variable $A$ dominates all bound variables, and $C$ has an atom set which is disjoint from all bound variables except $H$, which it dominates. So $Q$ satisfies conditions (1) and (2), and is thus p-hierarchical. Notice, however, that the sets $\at(A)$ and $\at(C)$ are neither disjoint, nor is one subsumed by the other one, so $Q$ is not q-hierarchical. 
\end{example}

We next present an alternative characterization of p-hierarchical queries in terms of their decomposition into q-hierarchical subqueries. 

\begin{proposition}
\label{prop-p-hierarchical-decomposition}
    A query $Q$ is p-hierarchical if and only if it can be expressed as a full conjunctive query with body atoms defined by free-dominated q-hierarchical subqueries.    
\end{proposition}

\begin{example}\label{ex:query-rewriting}
    The query $Q$ from Ex.~\ref{ex:p-hier-query} can be equivalently expressed as
    $Q(A,C) = Q_1(A) \cdot Q_2(A,C) \cdot Q_3(C)$,
    where the free-dominated q-hierarchical subqueries are defined as:
    \[
        \begin{array}{r@{\;\;=\;\;}l}
        Q_1(A)   & \sum_{B,D,E,F,G} R_1(A,B,D,E) \cdot R_2(A,B,D,F) \cdot R_3(A,B,G) \\
        Q_2(A,C) & \sum_H R_5(A,C,H) \\
        Q_3(C)   & R_4(C).
        \end{array}
    \]
\end{example}

The fractional hypertree width of a p-hierarchical query is equal to the fractional hypertree width of its bound version. This allows us to relate the complexity of maintaining a p-hierarchical query to the fractional hypertree width of its bound version (Theorem~\ref{thm:main-upper-bound}). 

\begin{proposition}
\label{prop:fhtw=fcfhtw}
    For any p-hierarchical query, its fractional hypertree width equals the fractional hypertree width of its bound version.
\end{proposition}

\section{Maintaining P-Hierarchical Queries}
\label{sec:upper-bound}

This section gives the maintenance algorithm behind Theorem~\ref{thm:main-upper-bound} and a proof sketch. By Prop.~\ref{prop-p-hierarchical-decomposition}, a p-hierarchical query $Q$ can be rewritten as a full conjunctive query whose atoms are free-dominated q-hierarchical subqueries; we show below how to construct this rewriting. Our maintenance structure follows the rewriting in two layers, each over a different semiring. Fig.~\ref{fig:create-var-orders-example} (right) shows the structure for the running example, with the subqueries $Q_1(A)$, $Q_2(A,C)$, and $Q_3(C)$ from Ex.~\ref{ex:query-rewriting}.

The bottom layer, below the dashed border, maintains the result of each subquery over the original semiring $K$; a hierarchy of materialized views, as in prior maintenance approaches for q-hierarchical queries~\cite{QHierarchical,DynYannakakis:SIGMOD:2017,FIVM:SIGMOD:2018,CROWN}, keeps the update time of each subquery constant. The top layer maintains only which tuples are present in the join of the subquery results, not their payloads from $K$: an indicator projection casts each such payload to the Boolean semiring. These indicator projections form the \emph{border} between the two layers. Their join $\Qtop(A,C)$ is a full conjunctive query over the Boolean semiring, which prior work~\cite{InsertsVSDeletes} maintains under inserts with the amortized update time and the enumeration delay stated in Theorem~\ref{thm:main-upper-bound}.

The border between the two layers keeps the update time low: it lets only unseen tuples through to the top layer. A tuple that does cross can be expensive to process, but its cost is amortized over the whole sequence of inserts. This behavior of the border is due to the change of semiring. Over $K$, the payload of a tuple already present in a subquery result can change repeatedly, with every insert. Over the Boolean semiring, the payload of a present tuple never changes, since $\mathsf{true} \lor \mathsf{true} = \mathsf{true}$: once a tuple has crossed the border, no later insert touches it.

The maintenance structure allows enumerating the result of $Q$ with the correct payloads. The result tuples are exactly those of the Boolean join of the subqueries, and the top layer enumerates them with constant delay. Their payloads, which the indicator projections do not carry, are recovered from the subquery results below the border, in constant time per tuple.

We proceed as follows: Sec.~\ref{subsec:sec5-prelims} recalls variable orders and view trees, Sec.~\ref{subsec:phase-1} builds the bottom layer, Sec.~\ref{subsec:phase-2} the top layer, and Sec.~\ref{subsec:proof-sketch} assembles the proof.

\begin{figure}[t]
    \centering
    \setlength{\tabcolsep}{14pt}
    \begin{tabular}{cc}
        \begin{tikzpicture}[baseline=(current bounding box.south),
            x=0.72cm, y=0.58cm,
            every node/.style={draw=none, inner sep=1.5pt, font=\scriptsize},
            ]
            \node (A1) at (1.15,4) {\underline{$A$}};
            \node (T1) at (0.4,4) {$\omega_1$};
            \node (B) at (1.15,3) {$B$};
            \node (D) at (0.5,2) {$D$};
            \node (G) at (1.8,2) {$G$};
            \node (E) at (0,1) {$E$};
            \node (F) at (1,1) {$F$};
            \node (R1) at (0,0) {$R_1$};
            \node (R2) at (1,0) {$R_2$};
            \node (R3) at (1.8,0) {$R_3$};
            \draw (A1) -- (B);
            \draw (B) -- (D);
            \draw (B) -- (G);
            \draw (D) -- (E);
            \draw (D) -- (F);
            \draw (E) -- (R1);
            \draw (F) -- (R2);
            \draw (G) -- (R3);
            \node (A2) at (3.4,3) {\underline{$A$}};
            \node (C2) at (3.4,2) {\underline{$C$}};
            \node (T2) at (2.75,3) {$\omega_2$};
            \node (H) at (3.4,1) {$H$};
            \node (R5) at (3.4,0) {$R_5$};
            \draw (A2) -- (C2);
            \draw (C2) -- (H);
            \draw (H) -- (R5);
            \node (T3) at (4.7,1) {$\omega_3$};
            \node (C3) at (5.3,1) {\underline{$C$}};
            \node (R4) at (5.3,0) {$R_4$};
            \draw (C3) -- (R4);
            \node at (0.9,-1) {$Q_1(A)$};
            \node at (3.4,-1) {$Q_2(A,C)$};
            \node at (5.3,-1) {$Q_3(C)$};
        \end{tikzpicture}
        &
        \begin{tikzpicture}[baseline=(current bounding box.south),
            x=0.87cm, y=0.64cm,
            every node/.style={draw=none, inner sep=1.5pt, font=\scriptsize},
            ]
            \node (R1) at (0,0) {$R_1(A,B,D,E)$};
            \node (R2) at (2.3,0) {$R_2(A,B,D,F)$};
            \node (V1) at (0,1) {$V_1(A,B,D)$};
            \node (V2) at (2.3,1) {$V_2(A,B,D)$};
            \node (V3) at (1.15,2) {$V_3(A,B,D)$};
            \node (R3) at (3.7,2) {$R_3(A,B,G)$};
            \node (V4) at (1.15,3) {$V_4(A,B)$};
            \node (V5) at (3.7,3) {$V_5(A,B)$};
            \node (V6) at (2.4,4) {$V_6(A,B)$};
            \node (V7) at (2.4,5) {$V_7(A)$};
            \node (I1) at (2.4,6.2) {$\mathbbm{1}_{V_7}(A)$};

            \node (R5) at (5.6,4) {$R_5(A,C,H)$};
            \node (V8) at (5.6,5) {$V_8(A,C)$};
            \node (I2) at (5.6,6.2) {$\mathbbm{1}_{V_8}(A,C)$};

            \node (R4) at (7.3,5) {$R_4(C)$};
            \node (I3) at (7.3,6.2) {$\mathbbm{1}_{R_4}(C)$};

            \node (QP) at (5.6,7.4) {$\Qtop(A,C)$};

            \draw (R1) -- (V1);
            \draw (R2) -- (V2);
            \draw (V1) -- (V3);
            \draw (V2) -- (V3);
            \draw (V3) -- (V4);
            \draw (R3) -- (V5);
            \draw (V4) -- (V6);
            \draw (V5) -- (V6);
            \draw (V6) -- (V7);
            \draw (V7) -- (I1);
            \draw (R5) -- (V8);
            \draw (V8) -- (I2);
            \draw (R4) -- (I3);
            \draw (I1) -- (QP);
            \draw (I2) -- (QP);
            \draw (I3) -- (QP);

            \draw[dashed, gray] (-0.9,5.6) -- (7.9,5.6);
            \node[gray, anchor=west] at (-0.9,6.15) {over $\mathbb{B}$};
            \node[gray, anchor=west] at (-0.9,5.05) {over $K$};
        \end{tikzpicture}
    \end{tabular}
    \caption{Left: the extended canonical variable orders $\omega_1$, $\omega_2$, and $\omega_3$ of the subqueries $Q_1(A)$, $Q_2(A,C)$, and $Q_3(C)$ from Ex.~\ref{ex:query-rewriting} of the query $Q$ from Ex.~\ref{ex:p-hier-query}. Free variables are underlined. The atoms $R_1,\ldots,R_5$ are shown without their schemas. Right: the maintenance structure for $Q$. Below the dashed border, there is one view tree per variable order, maintained over $K$. Their root views $V_7(A)$, $V_8(A,C)$, and $R_4(C)$ compute $Q_1$, $Q_2$, and $Q_3$. Above the border, the indicator projections and their join $\Qtop(A,C)$ are maintained over the Boolean semiring.}
    \label{fig:create-var-orders-example}
\end{figure}

\subsection{Variable Orders and View Trees}
\label{subsec:sec5-prelims}

We recall variable orders and view trees from prior work~\cite{FIVM,f-tree,H-IVM}.

\begin{definition}[Variable Order](Adapted from \cite[Def.~3.1]{FIVM})
    A variable order for a conjunctive query $Q$ is a rooted forest with one node per variable in $Q$ such that for each atom of $Q$, its variables lie along the same root-to-leaf path.
\end{definition}

A variable order is \emph{canonical} (adapted from~\cite{FIVM,IVMeps:PODS:2020}) if every variable dominates the variables below it, with the free variables on top (called free-top canonical in prior work~\cite{IVMeps:PODS:2020}, here we use a shortened name as we do not need canonical and not free-top variable orders). It is \emph{extended} if each atom is appended as a child of its lowest variable, so that the atoms are the leaves.

\begin{definition}[View Tree]~\cite[Def.~1]{H-IVM}
\label{def:viewtree}
A \emph{view tree} $T$ for a query $Q$ is a rooted tree with the properties:
\begin{itemize}
    \item There is a one-to-one mapping between the leaves of $T$ and the atoms of $Q$.
    \item Each inner node is a view over some variables of $Q$.
    \item If a node $V'(\vect{Y})$ has a single child node $V(\vect{X})$, then it is a \emph{projection view} defined by marginalizing variables of $V(\vect{X})$, i.e., $V'(\vect{Y}) = \sum_{\vect{X} \setminus \vect{Y}} V(\vect{X})$.
    Furthermore, every atom of $Q$ with a variable from $\vect{X} \setminus \vect{Y}$  occurs in the subtree rooted at $V(\vect{X})$.
    \item If a node $V(\vect{X})$ has several children $V_1(\vect{X_1}), \ldots , V_n(\vect{X_n})$, for $n \geq 2$, then it is a \emph{join view} defined by the natural join of the  child views, i.e., $V(\vect{X}) = V_1({\vect{X_1}})\cdot\ldots\cdot V_n(\vect{X_n})$.
\end{itemize}
\end{definition}

An insert into a base relation is propagated up its view tree: each view on the path from the leaf to the root is updated with the change of its child, a projection view by marginalizing it, a join view by joining it with the sibling views. 

F-IVM~\cite{FIVM} turns an extended canonical variable order of a q-hierarchical query into a view tree that maintains the query with constant update time: atoms are left unchanged, and each variable is replaced by a join view over its children, topped by a projection view in which the variable is marginalized out. Since the variable order is canonical, every atom below a variable contains that variable and all its ancestors, so every child view of a variable has exactly these variables as schema and the join view is an intersection. The change of an insert hence matches at most one tuple in each view along its path, and the update takes constant time.

\subsection{The Bottom Layer}
\label{subsec:phase-1}

\subparagraph{Decomposition.} We construct the decomposition of Prop.~\ref{prop-p-hierarchical-decomposition} together with an extended canonical variable order for each subquery; these variable orders guide the view tree construction below. A canonical variable order requires every variable to dominate the variables below it, with the free variables above the bound variables. These two requirements are matched by the two defining conditions for p-hierarchical queries (Def.~\ref{def:p-hierarchical}), and Algorithm~\ref{alg:create-var-orders} turns them into a construction.

Condition (1) states that two bound variables have either nested atom sets, so that one dominates the other, or disjoint atom sets. It lets us arrange the bound variables in a forest along the dominance order, with every bound variable below the bound variables that dominate it. This placement gives no variable two parents: two dominators of a bound variable $X$ share the atoms of $X$, so by condition (1) one of them dominates the other, and $X$ sits below the farther one through the closer one. Lines 1-4 therefore make every bound variable a child of its closest dominator, its \emph{minimal dominator}, or a root if it has none.

Condition (2) states that a free variable either dominates a bound variable or shares no atom with it, so a free variable that occurs in an atom of a tree dominates the root of the tree, which that atom contains. It lets us place the free variables that occur in the atoms of a tree above the bound variables of the tree. Line 8 places these free variables as a path above the root. Their order does not matter, since they all occur in every atom of the tree.

Finally, we attach the atoms. By condition (1), the bound variables of an atom lie in a single tree; lines 9-11 append the atom to that tree as a child of its lowest variable. An atom without bound variables lies in no tree, so lines 13-14 give it a variable order of its own, its variables as a path with the atom below. Each resulting variable order $\omega$ defines a subquery of $Q$: the join of the atoms of $\omega$, summed over the bound variables of $\omega$. The variable order $\omega$ is an extended canonical variable order of this subquery (Lemma~\ref{lem:creates-forest} in the appendix).

\begin{algorithm}[t]
\DontPrintSemicolon
\caption{\textsc{CreateVarOrders}($Q$)}
\label{alg:create-var-orders}

\KwIn{A p-hierarchical query $Q$, with its bound variables ordered by $X < Y$ if $\at(X) \subset \at(Y)$, or $\at(X) = \at(Y)$ and $X$ precedes $Y$ lexicographically}
\KwOut{Extended canonical variable orders of the subqueries of $Q$}
\BlankLine

Initialize a directed graph $G := (\mathcal{V}, \mathcal{E})$ with $\mathcal{V} := \Bound(Q)$ and $\mathcal{E} := \emptyset$\;

\ForEach{$X \in \mathcal{V}$}{
    \If{$\exists Y \in \mathcal{V} \textnormal{ s.t. } X < Y \land \not \exists Z \in \V \textnormal{ s.t. } X < Z < Y$}{
        $\mathcal{E} := \mathcal{E} \cup \{(Y \to X)\}$\tcp*{$X$ below its minimal dominator $Y$}
    }
}

$\Omega := \emptyset$\;

\ForEach{\textnormal{tree} $T$ \textnormal{of} $G$}{
    \textnormal{let} $\omega$ \textnormal{be a copy of} $T$ \textnormal{and let} $\mathcal{A} := \bigcup_{X \in \V(T)} \at(X)$\;
    \textnormal{add the free variables of} $Q$ \textnormal{in} $\mathcal{A}$ \textnormal{to} $\omega$ \textnormal{as a path above its root, in any order}\;
    \ForEach{$R(\vect Y) \in \mathcal{A}$}{
        \textnormal{let} $X$ \textnormal{be the lowest bound variable of} $\omega$ \textnormal{with} $X \in \vect Y$\;
        \textnormal{add} $R(\vect Y)$ \textnormal{to} $\omega$ \textnormal{as a child of} $X$\;
    }
    $\Omega := \Omega \cup \{\omega\}$\;
}

\ForEach{\textnormal{atom} $R(\vect Y)$ \textnormal{of} $Q$ \textnormal{with} $\vect Y \subseteq \Free(Q)$}{
    \textnormal{add to} $\Omega$ \textnormal{the path} $\vect Y$\textnormal{, in any order, with} $R(\vect Y)$ \textnormal{below it}\;
}

\Return $\Omega$\;
\end{algorithm}

\begin{example}\label{ex:bound-var-order-construction}
    We illustrate Algorithm~\ref{alg:create-var-orders} using $Q$ from Ex.~\ref{ex:p-hier-query} as input. Line 1 initializes the graph $G$ with the vertex set $\V = \Bound(Q) = \{B, D, E, F, G, H\}$. Consider the bound variable $E$ with $\at(E) = \{R_1(A,B,D,E)\}$. Both $B$ and $D$ dominate $E$, but $\at(E) \subset \at(D) \subset \at(B)$ gives $E < D < B$, so $D$ is the minimal dominator of $E$, and lines 2-4 add the edge from $D$ to $E$. The atom set $\at(H) = \{R_5(A,C,H)\}$ is disjoint from the atom sets of all other bound variables, so no variable dominates $H$ and it remains a root of its own. Processing the remaining variables in the same way yields two trees, one with root $B$ and one with the single node $H$. Line 8 places $A$ above $B$, and $A$ and $C$ above $H$. Lines 9-11 append the atoms: $R_1(A,B,D,E)$ below $E$, $R_2(A,B,D,F)$ below $F$, $R_3(A,B,G)$ below $G$, and $R_5(A,C,H)$ below $H$. The atom $R_4(C)$ has no bound variable, so lines 13-14 add $\omega_3$, with $C$ above $R_4(C)$. The result is the variable orders $\omega_1$, $\omega_2$, and $\omega_3$ in Fig.~\ref{fig:create-var-orders-example} (left).
\end{example}

The subquery of a variable order is free-dominated: its free variables dominate the root of the tree, hence occur in every atom of the subquery. It is hierarchical: two of its variables that share an atom lie on one root-to-leaf path, so one dominates the other. It has the q property: the free variables occur in all its atoms, so no variable has a strictly larger atom set than a free variable. Finally, a bound variable of $Q$ lies in exactly one tree, so the subqueries share only free variables of $Q$, and $Q$ is their full join:

\begin{proposition}\label{lem:q-hier}
    Given a p-hierarchical query $Q$, each variable order produced by Algorithm~\ref{alg:create-var-orders} is an extended canonical variable order of a free-dominated q-hierarchical subquery of $Q$, and $Q$ is the full join of these subqueries.
\end{proposition}

\begin{example}\label{ex:subqueries}
    Continuing Ex.~\ref{ex:bound-var-order-construction}, $\omega_1$ defines the subquery that joins $R_1$, $R_2$, and $R_3$ with the free variable $A$: this is $Q_1(A)$ from Ex.~\ref{ex:query-rewriting}, and $\omega_1$ is its extended canonical variable order. Likewise, $\omega_2$ defines $Q_2(A,C)$ and $\omega_3$ defines $Q_3(C)$. The three subqueries share only the free variables $A$ and $C$, and $Q(A,C) = Q_1(A) \cdot Q_2(A,C) \cdot Q_3(C)$, matching Ex.~\ref{ex:query-rewriting}.
\end{example}

\subparagraph{View trees.} We now build one view tree per subquery. Algorithm~\ref{alg:rewrite} (adapted from F-IVM~\cite{FIVM}, see Appendix~\ref{sec:appendix-upper-bounds}) applies the construction of Sec.~\ref{subsec:sec5-prelims} to each variable order with its free variables removed, so that only the bound variables become views. When a bound variable has a single child, the join view is redundant and can be omitted. The free variables produce no views: the projection view of the last bound variable already has the free variables of the subquery as its schema, so it holds the subquery result. It is the root view, on which the top layer joins. The variable order of an atom without bound variables becomes the atom itself, a trivial view tree.

\begin{example}\label{ex:view-tree-construction}
    Fig.~\ref{fig:create-var-orders-example} (right) shows, below the border, the resulting forest for the running example: the view tree with root $V_7(A)$ maintains $Q_1$, the view tree with root $V_8(A,C)$ maintains $Q_2$, and the trivial view tree $R_4(C)$ maintains $Q_3$. For instance, the variable $D$ of $\omega_1$ becomes the join view $V_3(A,B,D) = V_1(A,B,D) \cdot V_2(A,B,D)$, topped by the projection view $V_4(A,B) = \sum_{D} V_3(A,B,D)$. The variable $E$ has the single child $R_1(A,B,D,E)$, so its join view is omitted and it becomes only the projection view $V_1(A,B,D) = \sum_{E} R_1(A,B,D,E)$.
\end{example}

Each subquery is q-hierarchical and its variable order is extended canonical (Prop.~\ref{lem:q-hier}), so its view tree is updated in constant time, as in Sec.~\ref{subsec:sec5-prelims}. By Prop.~\ref{lem:q-hier}, $Q$ is the join of the root views of this forest:

\begin{proposition}\label{thm:main-alg}
    Given a p-hierarchical query $Q(\vect F)$, the construction above yields a forest of view trees $\T$, one per subquery of Prop.~\ref{lem:q-hier}, such that the root view of each tree holds the result of its subquery, hence $Q(\vect F) = \prod_{V(\vect Y) \in \roots(\T)} V(\vect Y)$, and each tree admits constant update time for a single-tuple insert.
\end{proposition}

\subsection{The Top Layer}
\label{subsec:phase-2}

Sec.~\ref{subsec:phase-1} leaves us with the forest $\T$ of view trees of Prop.~\ref{thm:main-alg}: its root views $V(\vect Y)$ hold the subquery results over $K$, and $Q(\vect F)$ is their join. The top layer maintains this join, but over the Boolean semiring. For each root view $V(\vect Y)$, let $\mathbbm{1}_{V}(\vect Y)$ denote its \emph{indicator projection}: the view over the Boolean semiring that maps a tuple to $\mathsf{true}$ exactly when its payload in $V(\vect Y)$ is nonzero. The indicator projections form the border between the two layers; in Fig.~\ref{fig:create-var-orders-example} (right), they are the views just above the dashed line. Above the border, let $\Qtop(\vect F) = \prod_{V(\vect Y) \in \roots(\T)} \mathbbm{1}_{V}(\vect Y)$ denote the join of the indicator projections, a full join query over the Boolean semiring.

The two conditions on $K$ make this change of semiring sound. The payload of a result tuple of $Q$ is the product of its payloads in the root views. By zero-divisor freeness, this product is nonzero exactly when all its factors are, so a tuple is in the result of $\Qtop$ if and only if its payload in $Q$ is nonzero; without it, a product of nonzero payloads could be zero, and $\Qtop$ would report tuples that $Q$ does not have. By zero-sum freeness, a payload that is nonzero stays nonzero under inserts, so the indicator projections only ever receive inserts; without it, an insert could cancel a payload back to zero, and the tuple would have to be deleted above the border, which the insert-only maintenance of the top layer cannot do.

Above the border, we maintain $\Qtop$ using prior work~\cite{InsertsVSDeletes}:

\begin{proposition}[{\cite[Theorem 4.1]{InsertsVSDeletes}}]\label{prop:top-layer}
    A full join query $\Qtop$ over the Boolean semiring can be maintained under a sequence of $N$ single-tuple inserts to an initially empty database with $\bigO(N^{\fhtw(\Qtop) - 1})$ amortized update time and constant enumeration delay.
\end{proposition}

This bound is in terms of $\Qtop$, the join whose atoms are the root views. The following lemma bounds $\fhtw(\Qtop)$ by $\fhtw(Q)$.

\begin{lemma}\label{lem:width-qprime}
    Let $Q$ be a p-hierarchical query and let $\Qtop$ be the join of the indicator projections over the root views of its forest of view trees. Then $\fhtw(\Qtop) \leq \fhtw(Q)$.
\end{lemma}

\subsection{Proof Sketch of Theorem~\ref{thm:main-upper-bound}}
\label{subsec:proof-sketch}

The proof of Theorem~\ref{thm:main-upper-bound} proceeds in three steps.

1. Below the border: by Prop.~\ref{thm:main-alg}, each of the $N$ single-tuple inserts updates one view tree of $\T$ in constant time and adds at most one tuple to its root view. The $N$ inserts hence translate into $\bigO(N)$ inserts into the relations of $\Qtop$.

2. Above the border: by Prop.~\ref{prop:top-layer}, these inserts are processed with $\bigO(N^{\fhtw(\Qtop) - 1})$ amortized update time each and $\Qtop$ is enumerated with constant delay, and $\fhtw(\Qtop) \leq \fhtw(Q)$ by Lemma~\ref{lem:width-qprime}, where $\fhtw(Q)$ equals the fractional hypertree width $\fhtw$ of the bound version of $Q$ by Prop.~\ref{prop:fhtw=fcfhtw}.

3. Enumeration: $Q$ and $\Qtop$ have the same result tuples, which can be enumerated with constant delay. While a tuple is emitted, its payload in $Q$ is recovered by projecting the tuple onto the schema of each root view of $\T$, looking up its payload there, and multiplying the payloads; these are at most $|\at(Q)|$ constant-time lookups, so the delay remains constant.

Altogether, $Q$ is maintained with $\bigO(N^{\fhtw - 1})$ amortized update time and constant enumeration delay, where $\fhtw$ is the fractional hypertree width of the bound version of $Q$, which proves Theorem~\ref{thm:main-upper-bound}.

{\noindent\bf Non-empty initial databases.} Theorem~\ref{thm:main-upper-bound} assumes an initially empty database. In practice, the database often already holds $M$ tuples when maintenance starts, and feeding them through the update procedure one at a time is inefficient. Instead, we preprocess the $M$ tuples in bulk and then maintain the result under the $N$ inserts. The bound of Theorem~\ref{thm:main-upper-bound} still holds: preprocessing the $M$ tuples and processing the $N$ inserts take $\bigO((N+M)^{\fhtw(Q) - 1})$ amortized time per tuple, over all $N + M$ tuples, and the enumeration delay stays constant.

The preprocessing computes every view directly from the $M$ tuples. Each view is a projection or an intersection and takes $\bigO(M)$ time. The structure that maintains $\Qtop$ consists of materialized views over a tree decomposition of $\Qtop$~\cite{InsertsVSDeletes} and takes $\bigO(M^{\fhtw(Q)})$ time by Lemma~\ref{lem:width-qprime}.
The preprocessed state is exactly the state that $M$ single-tuple inserts to an empty database would produce. The analysis of Theorem~\ref{thm:main-upper-bound} hence applies from this state on: it charges $\bigO((N+M)^{\fhtw(Q) - 1})$ time to each of the $N + M$ tuples, whether preprocessed or inserted. The preprocessing itself takes $\bigO(M^{\fhtw(Q)})$ time, that is, $\bigO(M^{\fhtw(Q) - 1})$ per preprocessed tuple, within the same charge. Together, this gives the claimed bound.

A further common use case keeps the state of the database and maintenance data structure after a sequence of inserts. A subsequent insert sequence can be processed directly on the current state as if it were part of the previous insert sequence and therefore with the same guarantees as in Theorem~\ref{thm:main-upper-bound}.

\nop{\subsection{Augmenting View Trees}

In the previous section, we showed how to decompose a p-hierarchical query $Q$ into q-hierarchical subqueries, and how to construct a forest of view trees $\T$ such that $Q$ is the join of the root views of these view trees. Now consider the query that results from this decomposition, $Q(\vect X) = \prod_{V(\vect Y) \in \roots(\T)} V(\vect Y)$. We can augment each view tree by adding an indicator projection $\mathbbm{1}_{V}(\vect Y)$ atop each root view $V(\vect Y)$. Joining these indicator projections yields the join query $\Qtop$ that can be maintained over the Boolean semiring using prior work~\cite{InsertsVSDeletes}, as long as we restrict ourselves to the insert-only setting. In particular, all tuples in $\Qtop$ have non-zero payload if and only if they have non-zero payload in $Q$. Thus, when we enumerate a tuple $t$ from $\Qtop$ with a non-zero payload, we can recover its payload in $Q$ by simply projecting each $t$ onto the schema of each indicator projection, and looking up its payload in the view below it in the forest of augmented view trees. Multiplying these payloads and replacing each payload of a tuple in $\Qtop$ with the respective product recovers the correct payloads of the tuples in $Q$, as if the entire query was maintained over the same semiring.}

\section{Lower Bounds for Non-P-Hierarchical Queries}
\label{sec:lower-bounds}

We now turn our attention to the lower bound result from Theorem~\ref{thm:main-lower-bound}. Its proof first shows that non-p-hierarchical queries over the natural and the tropical semiring cannot be maintained with $\bigO(N^{\frac{1}{2}-\gamma})$ amortized update time and $\bigO(N^{\frac{1}{2}-\gamma})$ enumeration delay for any $\gamma>0$, unless the OMv conjecture fails. It is then generalized to one of the two families of $\mathcal{K}$ (Def.~\ref{def:calK}): to the semirings with an injective homomorphism from the natural semiring, and to the idempotent strictly ordered semirings. The section presents each reduction on the query $Q_2$ of Ex.~\ref{ex:p-hier-query}; Appendix~\ref{sec:appendix-lower-bounds} gives the reductions for arbitrary non-p-hierarchical queries.

Our lower bounds rely on the hardness of the Online Matrix-Vector Multiplication (OMv) problem. 
The OMv conjecture is widely used to prove conditional lower bounds~\cite{QHierarchical, IVMeps:ICDT:2019,OMV:STOC:2015,DynamicWidth}.

\begin{definition}[Online Matrix-Vector Multiplication (OMv)~\cite{OMV:STOC:2015}]
    We are given an $n\times n$ Boolean matrix $\vect M$ and receive $n$ Boolean column vectors $\mathbf{v}_1, \dots, \mathbf{v}_n$ of size $n$, one by one; after seeing each vector $\vect v_i$, we output the product $\vect M\vect v_i$ before we see the next vector.
\end{definition}

It is strongly believed that the OMv problem cannot be solved in subcubic time.

\begin{conjecture}[OMv conjecture, Theorem 2.2~\cite{OMV:STOC:2015}]\label{conj:omv}
    For any $\gamma > 0$, there is no algorithm that solves this problem in time $\bigO(n^{3-\gamma})$.
\end{conjecture}

In the lower bound proofs, we reduce the OMv problem for an $n \times n$ Boolean matrix to the maintenance of a query under a sequence of $N = \bigO(n^2)$ single-tuple inserts to an initially empty $K$-database. If there were an algorithm maintaining a non-p-hierarchical query with $\bigO(N^{\frac{1}{2}-\gamma})$ amortized update time and $\bigO(N^{\frac{1}{2}-\gamma})$ enumeration delay, for some $\gamma>0$, we could use this algorithm to solve the OMv problem in time $\bigO(n^{3-2\gamma})$, which is subcubic.

\medskip

{\noindent \bf From Insert-Delete to Insert-Only.}
Prior lower bound proofs based on reductions from OMv were for the insert-delete setting, e.g.,~\cite{QHierarchical,IVMeps:ICDT:2019,DynamicWidth}. We recall their idea for $Q(X) = \sum_{Y} S(X,Y) \cdot T(Y)$, the simplest query that violates the bound-free interaction of Def.~\ref{def:p-hierarchical}. The reduction encodes the matrix into $S$ once, inserting $S(i,j) \mapsto 1$ for every $\vect M_{ij} = 1$. When the vector $\vect v_k$ arrives, it deletes the contents of $T$, inserts $T(j) \mapsto 1$ for every $\vect v_k[j] = 1$, and enumerates the query answers: the answer for $X = i$ is $\sum_{j} \vect M_{ij} \vect v_k[j] = (\vect M \vect v_k)[i]$, the $i$-th entry of the required product.

Our setting forbids the deletes: the reduction can only keep inserting into $T$, so the tuples corresponding to the earlier vectors remain, and the query answer mixes the current vector with all previous ones. Our reductions keep the inserts and neutralize the effect of the old tuples instead of deleting them. The natural and the tropical semiring offer two different mechanisms for this, which we present next. The bound-bound violation, on $Q_1$ of Ex.~\ref{ex:p-hier-query}, is handled by the same two mechanisms (Appendix~\ref{sec:appendix-lower-bounds}).

\nop{Prior lower bound proofs based on reductions from OMv/OuMv were for the insert-delete setting, e.g.,~\cite{QHierarchical,IVMeps:ICDT:2019,DynamicWidth}, whereas our proofs are for the more restricted insert-only setting. For the insert-delete setting, the OMv problem is encoded dynamically into the database: at each step, as a new vector $\vect v_k$ arrives, the reduction deletes the contents of the relation representing the previous vector and inserts the new one. Because our framework strictly prohibits deletions, we must encode the OMv steps differently. We bypass the need for deletions by continually inserting the exact same tuples into the database, but with carefully engineered payloads. However, because older tuples are not deleted, the relation's overall state becomes corrupted by prior insertions. To extract the correct answer for the current step, we must logically filter out this prior data using different algebraic mechanisms depending on the underlying semiring.}

\medskip

{\noindent \bf Natural Semiring: Accumulation and Differencing.}
Over the natural semiring, payloads accumulate additively: at step $k$ the relation $T$ represents the cumulative vector $\vect t_k = \sum_{m=1}^{k} \vect v_m$, and the answer for $X = i$ is $(\vect M \vect t_k)[i]$, the product over the cumulative vector instead of the current one. The required entry is the difference of two consecutive answers, $(\vect M \vect v_k)[i] = (\vect M \vect t_k)[i] - (\vect M \vect t_{k-1})[i]$, and the OMv output bit at $i$ is $1$ iff this entry is nonzero. The reduction\footnote{Xiao Hu previously presented a similar reduction for $Q() = \sum_{X,Y}R(X)\cdot S(X,Y)\cdot T(Y)$ and $Q(X) = \sum_Y S(X,Y)\cdot T(Y)$ under inserts-only and over $K$-databases where $K$ is  the natural semiring (\url{https://simons.berkeley.edu/talks/xiao-hu-university-waterloo-2023-09-29-0}). The proof in our paper was developed independently. After seeing an earlier draft of our work at the beginning of June 2026, Qichen Wang pointed us to Hu’s presentation.} thus stores the answers of the previous step and reports the bit as the result of the equality test $(\vect M \vect t_k)[i] \neq (\vect M \vect t_{k-1})[i]$.

This idea generalizes to every semiring $K$ with an injective homomorphism $h$ from the natural semiring (Def.~\ref{def:embedding}). The reduction over $K$ mirrors the reduction over $\N$: the inserts use payload $h(1)$ in place of $1$, and then, since $h$ preserves sums and products, every payload and every query answer is the element $h(m)$ in place of the value $m$. The recovery performs the same additions and equality tests on the answers over $K$; since $h$ is injective, the reduction over $K$ reports exactly the same bits (Appendix~\ref{sec:lowerbound-natural}).

\nop{In the natural semiring, payloads accumulate additively. At each step $k$, we insert tuples corresponding to the indices $i$ such that $\vect v_k[i] = 1$ with a payload of $1$. Because previous insertions cannot be deleted, the target relation represents the cumulative sum of all vectors seen so far, $\mathbf{t}_k = \sum_{i=1}^k \mathbf{v}_i$. Consequently, enumerating the query computes the matrix-vector product over this cumulative state, yielding $\mathbf{M}\mathbf{t}_k$. We bypass the lack of deletions by taking the algebraic difference between successive steps: subtracting the evaluated vector of step $k-1$ from step $k$ perfectly isolates the result for the current step, $\mathbf{M}\mathbf{v}_k$~\footnote{A proof similar in spirit for this reduction was presented by Xiao Hu on Sep 29, 2023 (\url{https://simons.berkeley.edu/talks/xiao-hu-university-waterloo-2023-09-29-0})}.}

\medskip

{\noindent \bf Tropical Semiring: Dominance via Minimization.}
Let $\mathbb{T} = (\R \cup \{\infty\}, \min, +, \infty, 0)$ denote the tropical semiring; its additive identity is $0_{\mathbb{T}} = \infty$ and its multiplicative identity is $1_{\mathbb{T}} = 0$. Over the tropical semiring, the sum is $\min$: stale payloads do not accumulate, and no differencing can remove them. The reduction instead encodes freshness in the payloads: the tuples inserted at step $k$ get the payload $\Delta_k$, where $\Delta_1 > \Delta_2 > \cdots$ is strictly decreasing (for instance $\Delta_k = -k$), so a fresher tuple always carries a strictly smaller payload. Concretely, the matrix is encoded with the multiplicative identity, $S(i,j) \mapsto 0$ for every $\vect M_{ij} = 1$, and step $k$ inserts $T(j) \mapsto \Delta_k$ for every $\vect v_k[j] = 1$. The answer for $X = i$ is $\min_{j} S(i,j) + T(j)$, which is the minimum of $T(j)$ over the indices $j$ with $\vect M_{ij} = 1$. If $(\vect M \vect v_k)[i] = 1$, some index $j$ has $\vect M_{ij} = 1$ and $\vect v_k[j] = 1$, so step $k$ inserts $T(j) \mapsto \Delta_k$ for this $j$. The payload $\Delta_k$ is the smallest so far, so the answer is exactly $\Delta_k$. If $(\vect M \vect v_k)[i] = 0$, then $\vect v_k[j] = 0$ for every index $j$ with $\vect M_{ij} = 1$, so step $k$ inserts none of these $T(j)$: each of them keeps a payload from an earlier step or stays absent, and the answer is at least $\Delta_{k-1} > \Delta_k$. The reduction thus reports the OMv output bit at $i$ as the result of the equality test against $\Delta_k$.

This idea generalizes to every idempotent strictly ordered semiring $K$ (Def.~\ref{def:temporal-chain-semiring}). First, a fresh payload must overwrite a stale one: for our chain, the sum of two payloads is the larger one (with respect to the natural order). Second, every step needs a payload that dominates all earlier ones: the infinite chain $0_K = c_0 \prec_K c_1 \prec_K \cdots$ supplies one per step. Third, we need to distinguish between a join in which a tuple still has a stale payload $(c_i\cdot c_{i+1})$ and a join where all tuples are annotated with fresh payloads $(c_{i+1}\cdot c_{i+1})$. The reduction translates the payloads, $1_K$ in place of $0$ and $c_k$ in place of $\Delta_k$ (Appendix~\ref{sec:lowerbound-tropical}).

\nop{In the tropical semiring, where the operators are minimum and addition, we cannot use subtraction to isolate the current step. Instead, we simulate state replacement by encoding ``freshness'' directly into the payloads. At each step $k$, we insert the indices $i$ such that $\vect v_k[i]=1$ with a strictly decreasing, negative payload $\Delta_k$ (where $\Delta_k < \Delta_{k-1}$). Because the semiring's summation acts as a minimum operator, the lowest-cost path through the joins will always be dominated by the most recently inserted tuples. Stale tuples from earlier steps retain larger payloads and are naturally filtered out by the minimization. By simply checking if the query's output exactly matches the expected combined payload of fresh tuples (for instance, $2\Delta_k$), we can determine if a valid OMv intersection exists at the current step, rendering the older data effectively invisible.}

\section{Conclusion}

In this paper, we study the role of semirings on the incremental view maintenance problem for conjunctive queries. We show that a conjunctive query without self-joins can be maintained under single-tuple inserts to an initially empty $K$-database with constant amortized update time and constant enumeration delay if and only if it is p-hierarchical $\alpha$-acyclic, for $K\in\mathcal{K}$. The ``only if'' result is conditional on the OMv and hyperclique conjecture. We also give a maintenance procedure for p-hierarchical queries and show that the amortized update time is $\bigO(N^{\fhtw-1})$, where $N$ is the length of the sequence of inserts and $\fhtw$ is the fractional hypertree width of the bound version of the query.

\nop{\haozheupdate{Our results assume an initially empty database. If the database is not empty at the start, its tuples can be inserted one at a time before the actual sequence of inserts, at the cost of one insert each; all bounds then hold with $N$ counting the initial tuples as well as the subsequent inserts.}}

There are two immediate lines of future work worth exploring:
\begin{itemize}
    \item Our understanding of optimality for IVM remains severely limited, as establishing optimality requires corresponding lower bounds. Recent work~\cite{Qichen:PODS:2026} introduced new conditional lower bounds for IVM beyond those conditioned on the OMv problem. A natural step is to define classes of queries whose maintenance complexity matches the new lower bounds for various semirings.

    \item One challenge brought by our maintenance algorithm for p-hierarchical queries concerns the enumeration of the query result: Whereas the entire query result can be enumerated with constant delay after each insert, it remains open whether the delta of the query result can also be enumerated with constant delay.
\end{itemize}

\bibliographystyle{plainurl}
\bibliography{bibliography}

\appendix
\section{Missing Preliminaries}
\label{sec:appendix-prelims}

{\noindent \bf Hypergraphs, Tree Decompositions, and Widths.} A \emph{hypergraph} $\H$ is a pair $\H=(\V,\E)$, where $\V$ is a set of vertices, and $\E\subseteq2^{\V}$ is a set of hyperedges; each hyperedge $Z\in\E$ is a subset of $\V$. Given a CQ $Q$, the \emph{hypergraph of $Q(\vect F)$} is a \emph{query hypergraph} $\H=(\V,\E,\vect F)$, where $\V:=\vars(Q)$ and $\E:=\set{Z\;|\;R(Z)\in\at(Q)}$. We identify the query $Q(\vect F)$ with its query hypergraph $\H$. For a subset $S\subseteq\V$, we use $Q[S] = (S,\E[S])$ to denote the subquery induced by $S$, where $\E[S] = \set{e\cap S\mid e\in\E}$. 

Given the query hypergraph $\H$ of $Q(\vect F)$, a \emph{tree decomposition} (TD) for $\H$ is a pair $(T,\chi)$, where $T$ is a tree, and $\chi:\nodes(T)\rightarrow2^{\V}$ such that:
\begin{itemize}
    \item for every hyperedge $Z\in\E$, there exists $t\in\nodes(T)$ with $Z\subseteq \chi(t)$;
    \item for every vertex $X\in\V$, the set $\set{t\in\nodes(T)\;|\;X\in \chi(T)}$ form a connected subtree of $T$. 
\end{itemize}

Each set $\chi(t)$ is called a \emph{bag} of the TD. A TD $(T,\chi)$ is \emph{free-connex} if there is a connected subtree $S$ of $T$ such that $\bigcup_{t\in\nodes(S)} \chi(t) = \vect F$, i.e.,~the union of the variables appearing in $S$ is exactly the set of free variables. Let $\FTD(\H)$ denote the set of free-connex TDs of $\H$. Note that for any conjunctive query $Q$, all the TDs of the bound version of $Q$ are trivially free-connex. 

A \emph{fractional edge cover} is a function $\rho : \E \to [0, 1]$ such that for every vertex $v \in \V$, the following condition holds (every vertex is covered):
\[
    \sum_{e \in \mathcal{E} : v \in e} \rho(e) \geq 1.
\]
The \emph{fractional edge cover number} of $Q$, denoted $\rho^*(Q)$, is the minimum total weight over all valid fractional edge covers:
\[
    \rho^*(Q) = \min_{\rho} \sum_{e \in \mathcal{E}} \rho(e).
\]
The \emph{width} of $(T,\chi)$ is defined as 
\[
\width(T,\chi) = \max_{t\in\nodes(T)} \rho^*(Q[\chi(t)]).
\]

\begin{definition}[Fractional Hypertree Width ($\fhtw$)]
    For any query $Q(\vect F) = (\V,\E)$, its fractional hypertree width $\fhtw$ is defined as:
    \[
    \fhtw(Q(\vect F)) = \min_{(T,\chi)\in\FTD(Q(\vect F))} \width(T,\chi).
    \]
\end{definition} 

We call a query $Q$ \emph{acyclic} iff there exists a tree decomposition for the hypergraph of $Q$ in which every bag is covered by an input atom. Finally, we present the definition of free-connex acyclic queries. 

\begin{definition}[Free-Connex Acyclic]
$Q$ is free-connex acyclic if $Q$ is acyclic and the modified query $Q^+$, constructed by adding a single new atom containing exactly the free variables of $Q$, remains acyclic.
\end{definition}

{\noindent \bf Semirings in Table~\ref{tab:semiring-table}.} Recall the covariance semiring $((\mathbb{N}, \mathbb{N}^m, \mathbb{N}^{m \times m}), +_K, \cdot_K)$, and let $\mathbf{D} = (\mathbb{N}, \mathbb{N}^m, \mathbb{N}^{m \times m})$. Then for $a = (c_a, \vect s_a, Q_a) \in \mathbf{D}$ and $b = (c_b, \vect s_b, Q_b) \in \mathbf{D}$, $a +_K b = (c_a + c_b, \vect s_a + \vect s_b, Q_a + Q_b)$ and $a \cdot_K b = (c_a c_b, c_b \vect s_a + c_a \vect s_b, c_b Q_a + c_a Q_b + \vect s_a \vect s_b^T + \vect s_b \vect s_a^T)$. In Table~\ref{tab:semiring-table}, it is stated that the covariance semiring over its full domain is not zero-divisor free. We will demonstrate this with an example and explain how this scenario does not occur in our context.

\begin{example} Let $a = (0, (0), (1))$ and $b = (0, (0),(2))$. Then $a \cdot_K b = (0, (0), (0)) = 0_K$, and yet $a, b \neq 0_K$, so the covariance semiring over its full domain is not zero-divisor free.
\end{example}
In previous database applications~\cite{FIVM:SIGMOD:2018,FIVM}, the payload of each tuple in a relation is mapped to the identity $1_K = (1, 0^{m \times 1}, 0^{m \times m})$. Without loss of generality, $\vars(Q)$ can be labeled $X_1, ..., X_{m}$. For each variable $X_j$ and for each $x \in \Dom(X_j)$, the lifting function is $g_{X_j}(x) = (1, \vect s, Q)$ where $\vect s$ is an $m \times 1$ vector of all zeros except $\vect s_j = x$, and $Q$ is a matrix of all zeros except $Q_{j,j} = x^2$. During the marginalization of $X_j$, each $X_j$-value $x$ is lifted and triples with the same key are summed up. In such applications, for a triple $(c, \vect s, Q) \in \mathbf{D}$, it can be inductively shown that $c = 0 \to \vect s = 0^{m \times 1} \land Q = 0^{m \times m}$. (This invariant holds initially and is preserved by lifting, addition, and multiplication.) For $a = (c_a, \vect s_a, Q_a) \in \mathbf{D}$ and $b = (c_b, \vect s_b, Q_b) \in \mathbf{D}$, if $a \cdot_K b = 0_K$, then $c_a \cdot c_b = 0$, and so $c_a = 0$ or $c_b = 0$, and $a = 0_K$ or $b = 0_K$, and so the covariance semiring in our context is zero-divisor free.

\begin{proposition}\label{prop:boolean-insert-only}
Over the Boolean semiring, an $\alpha$-acyclic conjunctive query without self-joins admits tractable maintenance under inserts if and only if it is free-connex, unless the Boolean matrix multiplication conjecture fails.
\end{proposition}

\begin{proof}
The upper bound is due to CROWN~\cite{CROWN}. For the lower bound, consider a non-free-connex $\alpha$-acyclic query $Q$ without self-joins, and suppose that $Q$ admits tractable maintenance under inserts. Given a static database $D$ of size $N$, we can construct the data structure for $Q(D)$ by inserting the $N$ tuples of $D$ into an initially empty database. This takes $\bigO(N)$ total time, after which the tuples in $Q(D)$ can be enumerated with constant delay. This contradicts the static lower bound of Bagan et al.~\cite{Free-connex:CSL:2007}, unless the Boolean Matrix Multiplication conjecture fails.
\end{proof}

\section{Proof of Corollary~\ref{cor:dichotomy}}

Theorem~\ref{thm:main-lower-bound} states that for any conjunctive query $Q$ without self-joins, and sequence of $N$ inserts to an initially empty $K$-database, where $K\in\mathcal{K}$, if $Q$ is not p-hierarchical, then for any $\gamma > 0$ there is no algorithm that maintains $Q$ with $\bigO(N^{\frac{1}{2}-\gamma})$ amortized update time and $\bigO(N^{\frac{1}{2}-\gamma})$ enumeration delay, unless the OMv conjecture fails. However, $Q$ can be cyclic, i.e., not $\alpha$-acyclic, and also p-hierarchical. Thus, we have to also argue that cyclic queries are not tractable. This follows from the following fact.  Assuming the Hyperclique conjecture, for a database $D$ and conjunctive query $Q$ with no self-joins, if $Q$ is cyclic, then $Q$ does not admit enumeration with $O(|D|)$ preprocessing time and $O(1)$ delay \cite[Theorem~3]{CheatersLemma}. This lower bound carries over to the dynamic insert-only setting: if $Q$ admitted amortized $O(1)$ time per insertion and $O(1)$ enumeration delay, then, given an arbitrary static database $D$, we could start from the empty database, insert all $|D|$ tuples of $D$ in total time $O(|D|)$, and subsequently enumerate $Q(D)$ with $O(1)$ delay, contradicting the static lower bound. Together with the upper bound stated in Theorem~\ref{thm:main-upper-bound}, this gives the required dichotomy. 

\section{Missing Details for Section~\ref{sec:p-hierarchical}}

In this section, we  present the proofs for the propositions given in Section~\ref{sec:p-hierarchical}.

\begin{proof}[Proof of Prop.~\ref{prop-p-hierarchical-decomposition}] The proof relies on Algorithm~\ref{alg:create-var-orders} and lemmas, which are introduced later in the appendix.

[$\Rightarrow$] By Lemma~\ref{lem:creates-forest}, given a p-hierarchical query $Q$, Algorithm~\ref{alg:create-var-orders} produces extended canonical variable orders. By Prop.~\ref{lem:q-hier}, each of them is the variable order of a free-dominated q-hierarchical subquery of $Q$, and $Q$ is the full join of these subqueries, which gives us the rewriting of $Q$ as desired.

[$\Leftarrow$] Consider a query $Q(\vect F)$ that is not p-hierarchical. Assume that $Q$ violates condition (1), that is, there are two bound variables $X,Y\in\Bound(Q)$ such that $\at(X)\cap\at(Y)\neq\emptyset$, $\at(X)\setminus\at(Y)\neq\emptyset$, and $\at(Y)\setminus\at(X)\neq\emptyset$. Consider a decomposition $Q(\vect F) = Q_1(\vect F_1)\cdot\ldots\cdot  Q_k(\vect F_k)$, where $Q_1,\dots,Q_k$ are all q-hierarchical and free-dominated. Since, $X,Y$ are bound variables of $Q$, there must be one subquery $Q_i(\vect F_i)$ that contains both $\at(X)$ and $\at(Y)$. However, this subquery cannot be q-hierarchical. 

Now, assume that $Q$ violates condition (2). Thus, there are variables $X\in\Bound(Q)$ and $Y\in\Free(Q)$ such that $\at(X)\cap\at(Y)\neq\emptyset$ and $\at(X)\setminus\at(Y)\neq\emptyset$. Again, consider a decomposition $Q(\vect F) = Q_1(\vect F_1)\cdot\ldots\cdot Q_k(\vect F_k)$, where $Q_1,\dots,Q_k$ are all q-hierarchical and free-dominated. There must be a subquery $Q_i(\vect F_i)$ that contains all the atoms in $\at(X)$. Since, $\at(X)\cap\at(Y)\neq\emptyset$ we have that $Y\in\vect F_i$. However, because $\at(X)\setminus\at(Y)\neq\emptyset$, $Q_i$ cannot be free-dominated. 
\end{proof}

\begin{proof}[Proof of Prop.~\ref{prop:fhtw=fcfhtw}]
    Let $Q(\vect F)$ be a p-hierarchical query. Any free-connex TD of $Q$ is also a free-connex TD of the bound version of $Q$. We argue that any optimal TD of the bound version of $Q$ can be transformed into a free-connex TD of $Q$ without increasing the fractional edge cover of any bag. 

    We start by decomposing $Q$ into free dominated q-hierarchical queries $Q(\vect F) = Q_1(\vect F_1)\cdot\ldots\cdot Q_k(\vect F_k)$ using Prop.~\ref{prop-p-hierarchical-decomposition}. Now, let $(T,\chi)$ be an optimal TD of $Q$ such that its width is exactly $w = \fhtw(Q)$. Consider the TD $(T',\chi')$, where $T'=T$ and $\chi'(t) = \chi(t)\cap\Free(Q)$. It follows that $\width(T',\chi')\leq w$ and $(T',\chi')$ still respects the connectedness property of TDs, but it might not cover all the atoms of $Q$. This tree will act as the required connected subtree that contains only free variables. 

    Now, for each subquery $Q_i(\vect F_i)$, we argue that $\fhtw(Q_i) \leq \fhtw(Q)$. Let $(T_i,\chi_i)$ be a TD, where $T_i = T$ and $\chi_i(t) = \chi(t)\cap \vars(Q_i)$. We show that the fractional edge cover of any bag $\chi_i(t)$ is at most $w$. Let $\rho$ be an optimal edge cover of bag $\chi(t)$, $t\in\nodes(T)$ with $\sum \rho(e) \leq w$. We aim to construct an edge cover $\rho'_i$ for the bag $\chi_i(t)$ that does not use edges $e\notin\at(Q_i)$. 
    
    By the definition of a p-hierarchical decomposition, any external atom $e\notin \at(Q_i)$ contains no bound variables of $Q_i$. Therefore, its contribution to the bag $\chi_i(t)$ is restricted entirely to the free variables $e\cap\vars(Q_i)\subseteq\Free(Q_i)$. Because $Q_i$ is a free-dominated subquery, every atom in $Q_i$ contains all of $\Free(Q_i)$. So, we can move that weight to any other edge that contains all the free variables. Pick an arbitrary edge $e^*\in\at(Q_i)$ such that $e^*$ covers $\vect F_i$. We define $\rho'$ exclusively over $\at(Q_i)$:
    \[
    \rho'(e) = 
    \begin{cases}
        \rho(e) + \sum_{e'\notin\at(Q_i)}\rho(e')& \text{if } e = e^*\\
        \rho(e) & \text{otherwise} 
    \end{cases}
    \]
    For any bound variable $X\in\Bound(Q_i)$, its coverage is preserved because it was only ever covered by atoms in $\at(Q_i)$. For any free variable $Y\in\Free(Q_i)$, the exact weight lost by discarding external atoms is transferred to $e^*$, which contains $Y$. Thus, $\rho'$ is a valid edge cover for $\chi_i'(t)$ with $\sum \rho'_i(e) \leq w$, so we conclude that $\fhtw(Q_i) \leq \fhtw(Q)$. 
    
    Furthermore, the TD $(T_i,\chi_i)$ must contain a bag with $\chi_i(t)\supseteq \vect F_i$. Using this node, we can attach the subtree $T_i$ to $T'$ by adding an edge between this node $t\in\nodes(T_i)$ to a node $t'\in\nodes(T')$ with $\chi'(t')\supseteq F_i$ while preserving the connectedness property of TDs for the free variables. The connectedness for the bound variables is guaranteed by each individual TD $T_i$. We do the same for all TDs to obtain a new TD $(T^*,\chi^*)$ for $Q$.
    
    The subtree $T'$ ensures the free-connex property and since every atom of $Q$ belongs to some subquery $Q_i$, it must be covered by some bag in $T_i$, so it is also covered by some bag in $T^*$. Finally, since every bag in $T'$ and $T_i$ has a fractional edge cover of at most $w$, the width of $(T^*,\chi^*)$ is at most $w$. Therefore, $(T^*,\chi^*)$ is a free-connex TD for $Q$ with $\width(T^*,\chi^*) = w$. 
\end{proof}

\section{Missing Details for Section \ref{sec:upper-bound}}
\label{sec:appendix-upper-bounds}

In this section, we provide the missing details for the construction of the view trees (Sec.~\ref{subsec:phase-1}). We also give the algorithm that transforms the variable orders into view trees. Then we give the proofs of correctness for our maintenance plan.

\subsection{Correctness of Algorithm~\ref{alg:create-var-orders}}

\begin{lemma}\label{lem:creates-forest}
    Given a p-hierarchical query $Q$, Algorithm~\ref{alg:create-var-orders} creates extended canonical variable orders over disjoint bound variables and atoms. Let $\omega_1, ..., \omega_n$ denote these variable orders with their free variables removed. In each $\omega_i$, the bound variables of every atom form a path that starts at the root, and the variable order corresponds to the free-dominated query that joins $\at(\vars(\omega_i))$.
\end{lemma}
\begin{proof}
    First we show that after line 4 of Algorithm~\ref{alg:create-var-orders}, the graph $G(\V, \E)$ is a forest. Suppose for the sake of contradiction that $G$ is not a forest. Then either (1) there exists a vertex with in-degree greater than one, or (2) there exists a directed cycle. Suppose case (1). Then $\exists Y,Z \in \V = \Bound(Q)$ that both have an edge to a vertex $X \in \V$, and without loss of generality, $Y$ precedes $Z$ lexicographically. This implies $\at(X) \subseteq \at(Y)$ and $\at(X) \subseteq \at(Z)$. If $\at(X) = \at(Y) \subseteq \at(Z)$, then $X < Y < Z$, so an edge would not exist from $Z$ to $X$. If $\at(X) = \at(Z) \subset \at(Y)$, then $X < Z < Y$, so an edge would not exist from $Y$ to $X$. Now consider the case that $\at(Y)$ and $\at(Z)$ are incomparable. Because $Q$ is p-hierarchical and $\at(Y) \cap \at(Z) \neq \emptyset$, at least one of the bound variables $Y$ or $Z$ must dominate the other, a contradiction. Now suppose case (2). Then $\exists X_1, ..., X_k \in \V$ that form a directed cycle $X_1, ..., X_k, X_1$. If we have an edge from $X_i$ to $X_j$, then $X_j < X_i$, so $\at(X_j) \subseteq \at(X_i)$. So it must be the case that all variables in this cycle have identical atom sets, and if we have an edge from $X_i$ to $X_j$ in this cycle, then $X_i$ is lexicographically larger than $X_j$. Because $X_1 < X_k < X_{k-1} ... < X_1$, $X_1$ is lexicographically larger than itself, a contradiction.
    
   

    Now we show that after line 4, for each tree $T$ in $G$, each atom in $\at(\V(T))$ (the atoms of $Q$ that contain a variable of $T$) has bound variables that lie along the same root-to-leaf path in $T$. Consider an atom $R(\vect Y) \in \at(\V(T))$ with bound variable $Y \in \vect Y$ s.t. $Y \in \V(T)$. Suppose for the sake of contradiction that there exists another bound variable $Z \in \vect Y$ which does not fall along the root to leaf path in $T$ which contains $Y$. Either (1) $Z \notin \V(T)$ or (2) $Z \in \V(T)$. Suppose the first case and let $W$ be the root of the tree $Z$ appears in. $\at(Y) \subseteq \at(X)$ and $\at(Z) \subset \at(Y)$, so $\at(X) \cap \at(Z) \neq \emptyset$ and because $Q$ is p-hierarchical, $X$ or $W$ must dominate the other. If $\at(X) \subset \at(W)$ or $\at(X) = \at(W)$ and $X$ precedes $W$ lexicographically, then $X < W$, and there should be an edge from $W$ (or the minimally dominating bound variable) to $X$, a contradiction, as $X$ is a root. The other case is symmetric. Thus, $Z \in \V(T)$. Now suppose case (2) and let $W$ be the first common ancestor of $Y$ and $Z$ in $T$. Then there is a path from $W$ to $Y$ and from $W$ to $Z$, so $Y < W$ and $Z < W$. Because $Q$ is p-hierarchical and $\at(Y) \cap \at(Z) \neq \emptyset$, $Y$ or $Z$ must dominate the other. If $\at(Y) \subset \at(Z)$ or $\at(Y) = \at(Z)$ and $Y$ precedes $Z$ lexicographically, then $Y < Z < W$. Let $Y_1$ be the first variable in the path from $W$ to $Y$ s.t. $Y_1 < Z$, and suppose $Y_2$ has an edge in this path to $Y_1$. Because $Y \leq Y_1 < Y_2$ and $Y_1 < Z$, $\at(Y_1) \subseteq \at(Z)$ and $\at(Y_1) \subseteq \at(Y_2)$, so $\at(Z) \cap \at(Y_2) \neq \emptyset$. Thus, $Y_1 < Z < Y_2$, and there does not exist an edge from $Y_2$ to $Y_1$, a contradiction.

    Next we show that the bound variables of $\vect Y$ form a subpath of a root-to-leaf path in $T$. Suppose for the sake of contradiction that they fall along the subpath $X_1, ..., X_i, X_{i+1}, ..., X_k$ of a root-to-leaf path, where $X_1, X_k \in \vect Y$ and $X_i \notin \vect Y$. Then $X_k < X_i$, so $\at(X_k) \subseteq \at(X_i)$. And yet $R(\vect Y) \in \at(X_k)$ and $R(\vect Y) \notin \at(X_i)$, a contradiction. 
    
    Next we show that this subpath must start at the root of $T$. Suppose for the sake of contradiction the bound variables of $\vect Y$ fall along a subpath of the root-to-leaf path $X_1, ..., X_i, X_{i+1}, ..., X_k$, where $X_1, ..., X_i \notin \vect Y$ and $X_{i+1} \in \vect Y$. Then $X_{i+1} < X_i$, so $\at(X_{i+1}) \subseteq \at(X_i)$, and yet $R(\vect Y) \in \at(X_{i+1})$ and $R(\vect Y) \notin \at(X_{i})$, a contradiction. Thus, for each atom in $\at(\V(T))$, the bound variables form a path starting at the root of $T$.

    Consider a leaf $Y$. Because $\at(Y) \subseteq \at(\V(T))$, all atoms $Y$ occurs in must have bound variables that form a subpath of a root-to-leaf path in $T$ starting at the root. Because $Y$ occurs in the schema of at least one atom $R(\vect Y)$ and $Y$ is a leaf, the bound variables of $\vect Y$ must form the path from the root to $Y$. Thus, each root-to-leaf path in $T$ is an atom of $Q$.
    
    Furthermore, the root $X$ dominates all bound variables $\V(T)$ in $T$, 
    and because $Q$ is p-hierarchical, a free variable that appears in a schema with a bound variable must dominate the bound variable. Thus, the variables in $\vars(\at(X)) \cap \Free(Q)$ appear in all atoms in $\at(\V(T))$, and the query corresponding to the join of $\at(\V(T)$ is free-dominated. Line 8 of Algorithm~\ref{alg:create-var-orders} adds exactly these free variables as a path on top of a copy of $T$, so the tree has one node per variable in $\at(\V(T))$, and each atom in $\at(\V(T))$ has a schema that forms a path beginning at the root. Thus, the tree is a canonical variable order of this free-dominated query, and lines 9-11 extend it with all atoms in $\at(\V(T))$, turning it into an extended canonical variable order where the leaves are precisely $\at(\V(T))$. Finally, each variable order added in lines 13-14 is the extended canonical variable order of a single atom of $Q$ without bound variables, which is a free-dominated query on its own.
\end{proof}

\begin{proof}[Proof of Prop.~\ref{lem:q-hier}]
    By Lemma~\ref{lem:creates-forest}, Algorithm~\ref{alg:create-var-orders} produces extended canonical variable orders; let $\omega_1, ..., \omega_n$ denote them with their free variables removed, where each $\omega_i$ corresponds to the free-dominated join query $Q_i(\vect F_i) = \prod_{R(\vect Y) \in \at(\vars(\omega_i))} R(\vect Y)$. Now we will show that if we marginalize out the bound variables $\vars(\omega_i)$, the resulting query $Q_i'$ which has bound variables $\vars(\omega_i)$, atoms $\at(\vars(\omega_i))$, and free variables $\vect F = \vars(\at(\vars(\omega_i))) \cap \Free(Q)$ is q-hierarchical. Consider any two bound variables $X_1, X_2 \in \vars(\omega_i)$. Because $Q$ is p-hierarchical, $\at(X_1) \subseteq \at(X_2)$, $\at(X_2) \subseteq \at(X_1)$, or $\at(X_1) \cap \at(X_2) = \emptyset$. Because $\at(X_1) \subseteq \at(\vars(\omega_i))$ and $\at(X_2) \subseteq \at(\vars(\omega_i))$, the property still holds for $Q_i'$. Consider a free variable $X_1 \in \vect F$ and a bound variable $X_2 \in \vars(\omega_i))$. Because $Q_i$ is free dominated (and $Q_i'$ as well, because it has the same atom set), $X_1$ appears in all atoms. Thus, $\at(X_2) \subseteq \at(\vars(\omega_i)) = \at(X_1)$, so all free variables dominate all bound variables in $Q_i'$, and all free variables dominate each other in $Q_i'$. Thus, $Q_i'$ is free-dominated q-hierarchical.

    Now we will show that $Q$ is the join of these q-hierarchical queries and the atoms of $Q$ which consist only of free variables. Because each atom of $Q_i$ has bound variables that form a path beginning at the root of $\omega_i$, and the $\omega_i$ have disjoint variables, they have disjoint atoms. Further, as each bound variable of $Q$ appears in some $\omega_i$, and each $\omega_i$ corresponds to the join of $\at(\vars(\omega_i)) = \leaves(Q_i)$, each atom with a bound variable in its schema appears as a leaf of some $\omega_i$. Additionally, each $Q_i$ is over distinct bound variables. Thus, each query $Q_i$ represents the join of disjoint atoms whose union is precisely the atoms of $Q$ that contain a bound variable, and each $Q_i$ has disjoint bound variables in its schema.

    Thus, for each $Q_i$, we can marginalize out the bound variables to obtain the query $Q_i'$. $Q$ is precisely the join of these queries and the atoms of $Q$ with schemas that contain only free variables; each such atom is the free-dominated q-hierarchical subquery whose variable order Algorithm~\ref{alg:create-var-orders} adds in lines 13-14.
\end{proof}

\subsection{View Tree Construction}

After constructing the variable orders (via Algorithm~\ref{alg:create-var-orders}), Algorithm~\ref{alg:rewrite} recursively transforms each of them, with its free variables removed, into a \emph{view tree} (Def.~\ref{def:viewtree}). The algorithm operates over a forest $\nu$ whose vertices are either bound variables, atoms, or views defined over their children. By processing the vertices in a post-order traversal, the algorithm applies the following logic: vertices corresponding to atoms are left unchanged (lines 4--5), while each bound variable vertex is replaced by a path consisting of a join view over its children topped by a projection view (lines 6--9) in which the bound variable is marginalized out. Ultimately, this bottom-up transformation gives a forest of view trees.

Note that when a bound variable node has only one child, the algorithm still creates both a join view and a projection view, even though no join is necessary. These redundant join views can be omitted, as shown in Fig.~\ref{fig:create-var-orders-example} (right).

\begin{algorithm}[t]
\DontPrintSemicolon
\caption{\textsc{Rewrite}($\nu$)}
\label{alg:rewrite}
\KwIn{A forest $\nu$}
\KwOut{A forest of view trees}
\BlankLine

\Switch{$\nu$}{
    \Case{$\{\nu_i\}_{i \in [n]}$}{
        \Return $\{\textsc{Rewrite}(\nu_i)\}_{i \in [n]}$\;
    }
    \Case{$R(\vect Y)$}{
        \Return $R(\vect Y)$\;
    }
    \Case{\protect\raisebox{-0.5\height}{\treeChildren}}{
        \textbf{let} $T_i = \textsc{Rewrite}(\nu_i)$ and let the root view be $V_i(\vect S_i),~\forall i \in [k]$\;
        \textbf{let} $\vect S = \bigcup_{i \in [k]} \vect S_i$\;
        \Return $
        \begin{cases}
            \treeChildrenReturn
        \end{cases}
        $\;
    }
}
\end{algorithm}

\begin{lemma}\label{lem:view-trees}
    Given the variable orders produced by Algorithm~\ref{alg:create-var-orders} with their free variables removed, Algorithm~\ref{alg:rewrite} produces a view tree for each of them. Additionally, the root of each view tree has a schema that consists precisely of the free variables of the respective query, and each join view corresponds to an intersection.
\end{lemma}
\begin{proof}
    Let the original input to Algorithm~\ref{alg:rewrite} be the forest $\omega = \{\omega_1, ..., \omega_n\}$ of the variable orders with their free variables removed, where each $\omega_i$ corresponds to the join query $Q_i$. In lines 2-3, we recursively invoke Algorithm~\ref{alg:rewrite} for each $\omega_i \in \omega$, and we visit each variable of $\omega_i$ in a postorder traversal (lines 6-7). By construction of each $\omega_i$, the leaves are precisely $\at(Q_i)$, which are untouched by Algorithm~\ref{alg:rewrite} (lines 4-5). Each inner node of $\omega_i$ becomes a view with a schema defined over its children (lines 8-9). Inductively, each view is over variables of $Q_i$. 

    When Algorithm~\ref{alg:rewrite} visits the variable $X$, we replace it with two connected views: a view $V(\vect S)$ that is the join of the children views, and a marginalization view $V(\vect S \setminus X)$ that is the parent of $V(\vect S)$ (line 9). As all atoms of $Q_i$ that contain the bound variable $X$ occur in the subtree rooted at $V(\vect S)$ (by construction of $\omega_i$), Algorithm~\ref{alg:rewrite} produces a view tree $T_i$ for each $\omega_i$.

    Furthermore, as all atoms of $Q_i$ are leaves of $T_i$, and $T_i$ only consists of join and projection views (that project out only bound variables), the root view of $T_i$ is over the free variables of $Q_i$.

    Now we inductively show that when Algorithm~\ref{alg:rewrite} visits the inner node $X$ in a postorder traversal of $\omega_i$, all of the children of $X$ are views with the same schema $\pathh(X)$, and so $X$ is replaced by an intersection view with schema $\pathh(X)$, which has a parent view with the schema $\pathh(X) \setminus X$ (line 9). As the base case, consider visiting the node $X_1$ in $\omega_i$, where the children of $X_1$ are all leaves (atoms of $Q_i$). By construction of $\omega_i$, all leaves of $X_1$ are atoms with the same schema, $\pathh(X_1)$. So Algorithm~\ref{alg:rewrite} replaces $X$ with the intersection view with schema $\pathh(X_1)$, and with a parent view with schema $\pathh(X_1) \setminus X_1$. As the induction step, consider visiting an inner node $X_k$, and assume the claim holds for all inner nodes that have already been visited in the postorder traversal. Because $X_k$ is visited only after all of its children have been visited, the inductive hypothesis applies to all of its children. Thus, each of its children $X_j$ in $\omega_i$ has been replaced with an intersection view that has schema $\pathh(X_j)$, with a parent that has schema $\pathh(X_j) \setminus X_j = \pathh(X_k)$. Thus, all children of $X_k$ are projection views with the same schema $\pathh(X_k)$, and $X_k$ is replaced with an intersection view with schema $\pathh(X_k)$, with a parent that is the projection view with schema $\pathh(X_k) \setminus X_k$. As the claim holds for $X_k$, we have inductively shown all inner nodes are replaced by intersection and projection views, and so all join views correspond to intersections.
\end{proof}

\begin{proof}[Proof of Prop.~\ref{thm:main-alg}]
    By Lemma~\ref{lem:creates-forest}, given a p-hierarchical query $Q$, Algorithm~\ref{alg:create-var-orders} constructs extended canonical variable orders of subqueries, and by Prop.~\ref{lem:q-hier}, $Q$ is the join of the corresponding queries (if we project out the bound variables of $Q$). By Lemma~\ref{lem:view-trees}, Algorithm~\ref{alg:rewrite} takes these variable orders, with their free variables removed, as input, and produces a view tree for each of them. Additionally, the root of each view tree has a schema that consists precisely of the free variables of the respective query, and each join view corresponds to an intersection. Thus, executing Algorithm~\ref{alg:create-var-orders} with $Q$ as input, and using the resulting variable orders, with their free variables removed, as input to Algorithm~\ref{alg:rewrite}, yields a forest of view trees $\T$ s.t. $Q$ is the join of the root view of each view tree. Further, because each view tree consists only of projection and intersection views (or stand-alone relations), and intersections and projections and relations each take constant time to update (given a single-tuple insert), each view tree in $\T$ admits constant update time given a single-tuple update, and $Q$ is precisely is the join of of the root views of each view tree in the forest $\T$ produced by our maintenance approach.
\end{proof}

\subsection{Proof of Theorem~\ref{thm:main-upper-bound}}

\begin{proof}[Proof of Lemma~\ref{lem:width-qprime}]
    Let $\H=(\V,\E,\vect F)$ denote the hypergraph of $Q$, and let $\H'=(\V',\E',\vect F)$ denote the hypergraph of $\Qtop$. Let $(T, \chi)$ be a minimum-width free-connex tree decomposition of $Q$. Consider the tree decomposition $(T', \chi')$ obtained by removing the bound variables of $Q$ from each bag of the tree decomposition $(T, \chi)$. We will show $(T', \chi')$ is a valid tree decomposition for $\Qtop$. By Prop.~\ref{lem:q-hier} and Lemma~\ref{lem:view-trees}, $\T$ is a forest of view trees corresponding to free-dominated subqueries such that each root view is over the free variables of the respective subquery. Consider an edge $e' \in \E'$. The edges of $E'$ are precisely the set of schemas of $\roots(\T)$, so $e'$ is the schema of the root view of a view tree corresponding to a subquery $Q_i$ that is free-dominated. Thus, there exists an atom of $Q_i$ with schema $e \in \E$ s.t. $e' \subseteq e$. Because $(T, \chi)$ is a tree decomposition of $Q$, there exists a vertex $t \in \nodes(T)$ s.t. $e \subseteq \chi(t)$. Thus, $e' \subseteq \chi(t)$. Additionally, because $e' \subseteq \vect F$, and we only remove bound variables, $e' \subseteq \chi'(t)$. Thus, for every edge $e' \in \E'$, $e'$ is a subset of $\chi'(t)$ for some $t \in \nodes(T')$, and so every variable of $\Qtop$ appears in a bag of $(T', \chi')$.
    
    Now we will show that for every variable $X \in \V' = \vect F$, the set $\{t | X \in \chi'(t)\}$ is a connected subtree in $T'$ (we have already shown the set is non-empty). Because $X \in \V$, the set $\{t | X \in \chi(t)\}$ is a non-empty connected subset in $T$. Now suppose for the sake of contradiction that the set is disconnected in $T'$. Then there exists a connected component $C_1$ with vertex $t_1 \in \nodes(T')$ and a connected component $C_2$ with vertex $t_2 \in \nodes(T)'$ s.t. $X \in \chi'(t_1)$ and $X \in \chi'(t_2)$. Additionally, $t_1, t_2 \in \nodes(T)$, and $t_1$ and $t_2$ are connected in $T$. Consider the path from $t_1$ to $t_2$ in $T$. Then there exists a vertex $t_j$ along this path s.t. $X \in \chi(t_j)$ and $X \notin \chi'(t_j)$. This is a contradiction, because we only remove bound variables, and $X$ is free. Thus, for all $X \in \V'$, $\{t|X \in \chi'(t)\}$ forms a non-empty connected subtree in $T'$, and so $(T', \chi')$ is a valid tree decomposition of $\Qtop$.

    Now we will show the width of $(T', \chi')$ is less than or equal to the width of $(T, \chi)$. Consider a vertex $t \in \nodes(T) = \nodes(T')$. Then $\chi'(t) \subseteq \chi(t)$. Consider an optimal fractional edge cover $\rho_t$ of $\chi(t)$. Let $\rho'_t$ map each $e' \in \E'$ to zero. Let $e \in \E$ be an edge of $\H$ that contains a free variable. Because $\T$ is a forest of view trees s.t. each query is free-dominated and $Q$ is the join of the root views (whose schemas are the edges $\E'$) there exists an edge $e' \in \E'$ s.t. $e' = e \cap \vect F$. Increment $\rho'_t(e')$ by $\rho_t(e)$, and repeat this for all edges of $\E$ that contain a free variable. In the end, $\rho'_t$ is a fractional edge cover for $\chi'(t)$ (because $\chi'(t)$ contains only free variables), and $\rho^*(\Qtop[\chi'(t)]) \leq \sum_{e \in \E'[\chi(t)]} \rho'_t(e) \leq \sum_{e \in \E[\chi(t)]} \rho_t(e) = \rho^*(Q[\chi(t)])$. Further, because this holds true for all vertices $t \in \nodes(T)$, we have 
    \begin{align*}
        \fhtw(\Qtop) &= \min_{(T_1,\chi)\in\TD(Q(\vect F))} \width(T_1,\chi) \\
        & \leq \width(T',\chi) \\
        & = \max_{t\in\nodes(T')} \rho^*(Q[\chi(t)]) \\
        & \leq \max_{t \in \nodes(T')} \sum_{e \in \E'[\chi'(t)]} \rho'_t(e) \\
        & \leq \max_{t \in \nodes(T')} \sum_{e \in \E[\chi(t)]} \rho_t(e) \\
        & = \max_{t \in \nodes(T)} \rho*(Q[\chi(t)]) \\
        & = \width(T, \chi) \\
        & = \fhtw(Q).
    \end{align*}
\end{proof}

We are now ready to present the proof of the main theorem.

\begin{proof}[Proof of Theorem~\ref{thm:main-upper-bound}]
    Consider the forest of view trees $\T$ produced by our maintenance approach for the p-hierarchical query $Q$. By Prop.~\ref{thm:main-alg}, $Q(\vect F) = \prod_{V(\vect Y) \in \roots(\T)} V(\vect Y)$. Consider the forest of view trees $\T'$ obtained by adding an indicator projection $\mathbbm{1}_{V}(\vect Y)$ atop each root view $V(\vect Y)$ in $\T$. Then all views in the set $\roots(\T')$ are over the Boolean semiring, and the query $\Qtop(\vect F) = \prod_{V(\vect Y) \in \roots(\T')} V(\vect Y)$ maps all tuples in $Q$ with non-zero payload to the payload one, and zero otherwise.
    
    Consider a single-tuple insert to $\D$. Because $\T$ has disjoint atoms (Lemma~\ref{lem:creates-forest}), only one view tree $T$ in the forest $\T$ will be updated. By Prop.~\ref{thm:main-alg}, $T$ admits constant update time. This implies the root view $V(\vect Y)$ of $T$ admits constant update time, and so the indicator projection $\mathbbm{1}_{V}(\vect Y)$ admits constant update time as well. As the view trees in $\T'$ are precisely those in $\T$ with an indicator projection on top, the augmented view tree $T'$ in $\T'$ that corresponds to $T$ in $\T$ also admits constant update time. Thus, given a single-tuple insert to $\D$, $\T'$ admits constant update time.
    
    Let $\D$ be an initially empty database (so all views in $\T'$ are initially empty), and consider $N$ single-tuple inserts to $\D$. Because $\T'$ admits constant update time for each single-tuple insert and only one tree in $\T'$ is updated, its root admits constant update time and thus the number of tuples inserted into the root is $\bigO(1)$. Summing up over $N$ single-tuple inserts implies $\bigO(N)$ total inserts into the roots of $\T'$. As these root views are precisely the input relations of $\Qtop$, and $\Qtop$ is a full-join query over the Boolean semiring, $\Qtop$ admits $\bigO(N^{\fhtw(\Qtop) - 1})$ amortized update time and non-amortized constant enumeration delay~\cite[Theorem 4.1]{InsertsVSDeletes}. 

    Now we will show how to use the result of $\Qtop$ to obtain the result of $Q$. Suppose we perform a single-tuple insert to $\D$, update $\T'$, and perform all inserts to input relations of $\Qtop$ as defined by $\roots(\T')$. It is immediate that $\Qtop$ maps all tuples in $Q$ with non-zero payload to the payload one, and zero otherwise. Consider a tuple $\vect x \in \Dom(\vect X)$ s.t. $Q(\vect x) \neq \mathbf{0}$. Then its payload is $Q(\vect x) = \prod_{V(\vect Y) \in \roots(\T)} V(\pi_{\vect Y} \vect x)$. Thus, for each tuple $\vect x \in \Dom(\vect X)$ s.t. $\Qtop(\vect x) \neq \mathbf{0}$, we can project $\vect x$ onto the schema of each indicator projection in $\T'$ (the forest of augmented view trees), look up the payload of this projected tuple in the child view, and take the product of these payloads. This product is precisely the payload $Q(\vect x)$. Further, each projection and look-up takes constant time, and the number of times we perform each operation is bounded by $|\at(Q)|$, which is constant w.r.t. data complexity.
    
    Thus, when we begin with an initially empty database $\D$ and receive a stream of $N$ single-tuple inserts, after each insert, we update the forest of augmented view trees $\T'$ (which takes constant time) and we update the query $\Qtop$ (which admits $\bigO(N^{\fhtw(\Qtop)-1})$ amortized update time and constant enumeration delay). When enumerating the output of $\Qtop$, each tuple in $\Qtop$ with non-zero payload is enumerated with constant enumeration delay, and while enumerating the tuple, we recover the payload in $Q$ in constant time by multiplying payloads of views in $\T'$, as described above. Thus, each insert into $\D$ takes $\bigO(N^{\fhtw(\Qtop)-1})$ amortized update time and admits non-amortized constant delay enumeration of the output.
    
    By Lemma~\ref{lem:width-qprime}, $\fhtw(\Qtop) \leq \fhtw(Q)$. Thus, each insert takes $\bigO(N^{\fhtw(Q)-1})$ amortized update time. By Prop.~\ref{prop:fhtw=fcfhtw}, $\fhtw(Q)$ equals the fractional hypertree width of the bound version of $Q$, which proves Theorem~\ref{thm:main-upper-bound}.
\end{proof}

\section{Missing Details for Section~\ref{sec:lower-bounds}}
\label{sec:appendix-lower-bounds}

The OMv problem and the OMv conjecture are stated in Sec.~\ref{sec:lower-bounds}. Although we base our results on the hardness of the OMv conjecture, we introduce the following additional conjecture which is a consequence of the OMv conjecture~\cite{OMV:STOC:2015} since it is a more suitable starting point for some of our reductions. 

\begin{definition}[Online Vector-Matrix-Vector Multiplication (OuMv)~\cite{OMV:STOC:2015}]
    We are given an $n\times n$ Boolean matrix $\vect M$ and receive $n$ pairs of Boolean column vectors $(\vect u_1,\mathbf{v}_1), \dots, (\vect u_n,\mathbf{v}_n)$ of size $n$, one by one; after seeing each pair of vectors $(\vect u_i,\vect v_i)$, we output the product $\vect u_i^\top\vect M\vect v_i$ before we see the next pair. 
\end{definition}

\begin{conjecture}[OuMv conjecture, Theorem 2.4~\cite{OMV:STOC:2015}]
    For any $\gamma > 0$, there is no algorithm that solves this problem in time $\bigO(n^{3-\gamma})$.
\end{conjecture}

Enumeration returns only the tuples with nonzero payload (Sec.~\ref{sec:prelims}). The reductions below therefore read an empty enumeration result as the answer $0_K$, and an output tuple that is not enumerated as having payload $0_K$.

\subsection{The Natural Semiring}
\label{sec:lowerbound-natural}

We first show the lower bound for the natural semiring $(\N,+,\cdot)$.

Let $Q$ be any conjunctive query without self-joins that is not p-hierarchical. We reduce the OMv problem to the incremental maintenance of $Q$. The query $Q$ must violate one of the two conditions from Def.~\ref{def:p-hierarchical}. We next analyze each of these two cases in isolation.
    
    \subsubsection{Case 1} Assume $Q$ violates the bound-bound interaction, that is, $Q$ is of the form $Q(\vect F) = \sum_{\vect B}R(X,\vect Z_R)\cdot S(X,Y,\vect Z_s)\cdot T(Y,\vect Z_T)\cdot U_1(\vect Z_1)\cdot\ldots\cdot U_m(\vect Z_m)$, subject to the constraints $X,Y \notin \vect F$, $X \notin \vect Z_T$, and $Y\notin \vect Z_R$. Indeed, the violation gives two bound variables $X$ and $Y$ with $\at(X) \not\subseteq \at(Y)$, $\at(Y) \not\subseteq \at(X)$, and $\at(X) \cap \at(Y) \neq \emptyset$: let $S$ denote an atom in $\at(X) \cap \at(Y)$, $R$ an atom in $\at(X) \setminus \at(Y)$, $T$ an atom in $\at(Y) \setminus \at(X)$, and $U_1, \ldots, U_m$ the remaining atoms.
    
    \paragraph*{Database Construction.} We simplify the query $Q$ to its hard subquery $Q'() = \sum_{X,Y} R(X)\cdot S(X,Y)\cdot T(Y)$ by fixing the variables in $\vars(Q)\setminus\set{X,Y}$ to a special constant $\star$ that is not in the domains of $X$ and $Y$. For conciseness, we overload $\star$ to also represent a tuple consisting entirely of this constant, with arity determined by a relation's schema. To neutralize the effect of the relations $U_1, \dots, U_m$ on the joins on $X$ and $Y$, we execute the following inserts into each $U_i(\vect Z_i)$:
    \begin{itemize}
        \item If $Z_i \cap \set{X,Y} = \emptyset$, we insert $U_i(\star)\mapsto1$.
        \item If $Z_i \cap \set{X,Y} = \set{X}$ or $Z_i \cap \set{X,Y} = \set{Y}$, we insert $U_i(j,\star)\mapsto1$ for all $j \in [n]$, where $j$ is the value for $X$ or $Y$.
        \item If $Z_i \cap \set{X,Y} = \set{X,Y}$, we insert $U_i(j,k,\star)\mapsto 1$ for all $(j,k) \in [n]^2$, where $j,k$ are the values for $X,Y$.
    \end{itemize}
    Note that $1$ is the multiplicative identity of the natural semiring.
    By construction, $Q(\star)$ and $Q'()$ are equivalent for the above-constructed database with arbitrary $\mathbb{N}$-relations  $R$, $S$, and $T$.
        
    Let $\mathbf{M}$ be the $n \times n$ Boolean matrix in the OMv problem. We encode $\mathbf{M}$ into the $\mathbb{N}$-relation $S$ using at most $n^2$ inserts: for every $\mathbf{M}_{ij}=1$, we insert $S(i,j,\star)\mapsto1$. 
    The OMv problem also gives a sequence of vectors $\mathbf{v}_1, \dots, \mathbf{v}_n$. At step $k \in [n]$, we must compute the vector $\mathbf{M} \mathbf{v}_k$. Upon receiving $\mathbf{v}_k$, we make the following inserts into relation $T$. For every index $j$ such that $\mathbf{v}_k[j]=1$, we insert $T(j,\star)\mapsto1$. For all other indices $j$, i.e., indices $j$ for which $\mathbf{v}_k[j]=0$, we make no inserts. 

    A key constraint of our reduction is that we cannot delete from the database relations. Yet we can achieve the effect of deletes followed by inserts by inspecting the difference in the query answer before and after inserts. This is explained next.
    Since we cannot delete the previous entries in $T$ and because the natural semiring accumulates payloads, the vector represented by $T$ at step $k$ is the cumulative vector $\mathbf{t}_{k} = \sum_{m=1}^k \mathbf{v}_m$. This sum is over the natural semiring, not the Boolean semiring. Even though, $\vect v_1,\dots,\vect v_n$ are Boolean vectors, we treat them as integer vectors and recover the Boolean result out of the query answer. The database implicitly holds the state $\mathbf{u}_{k} = \mathbf{M} \mathbf{t}_{k}$. 
    
    \paragraph*{Recovering the Result for the OMv Problem from the Query Result.}
    To extract the $n$-dimensional vector $\mathbf{u}_{k}$ from the scalar result of $Q(\star)$, we execute a sequence of $2n$ query enumerations interleaved with $n$ inserts. For each index $i \in [1, n]$:
    \begin{enumerate}
        \item We enumerate the current scalar output $Q_{\text{before}}^{(i)} = Q(\star)$.
        \item We insert $R(i,\star)\mapsto1$. 
        \item We enumerate the new scalar output $Q_{\text{after}}^{(i)} = Q(\star)$.
    \end{enumerate}
    Let $R'(i)$ be the value of $R(i)$ before we inserted $R(i,\star)\mapsto1$. Inserting $R(i,\star)\mapsto 1$ increases the scalar query answer by exactly the value given by the dot product between the $i$-th row of $\vect M$ and the vector $\vect t_{k}$:
    \begin{align*}
        Q_{\text{after}}^{(i)} - Q_{\text{before}}^{(i)} &=
        \left(\sum_{x\in [n]\setminus\set{i},y} R(x)\cdot S(x,y)\cdot T(y) + \sum_{y} R(i)\cdot S(i,y)\cdot T(y)\right) \\
        &- \left(\sum_{x\in [n]\setminus\set{i},y} R(x)\cdot S(x,y)\cdot T(y) + \sum_{y} R'(i)\cdot S(i,y)\cdot T(y)\right)\\
        &=\sum_y \underbrace{(R(i)-R'(i))}_{1}\cdot S(i,y)\cdot T(y)
        =\sum_y S(i,y)\cdot T(y) =  \mathbf{M}_i \mathbf{t}_{k}  = \mathbf{u}_{k}[i].
    \end{align*}
    Let $a_k^{(i)}$ and $b_k^{(i)}$ denote the answers $Q_{\text{after}}^{(i)}$ and $Q_{\text{before}}^{(i)}$ enumerated at step $k$, so $\vect u_k[i] = a_k^{(i)} - b_k^{(i)}$, and let $a_0^{(i)} = b_0^{(i)} = 0$.

    Since $\vect t_k - \vect t_{k-1} = \vect v_k$ (with $\vect t_0 = \vect 0$), the $i$-th entry of $\vect M \vect v_k$ is $\vect u_k[i] - \vect u_{k-1}[i]$. This entry is nonnegative, so it is nonzero iff $\vect u_k[i] \neq \vect u_{k-1}[i]$, that is, iff $a_k^{(i)} + b_{k-1}^{(i)} \neq a_{k-1}^{(i)} + b_k^{(i)}$. At each step, the reduction stores the $2n$ enumerated answers, keeps those of the previous step, and outputs for each $i \in [n]$ the Boolean result of this equality test. The recovery thus uses only additions and equality tests on enumerated answers.

    \paragraph*{Time Analysis.} Starting from the empty database, the reduction executes at most $n^2$ inserts to encode $\mathbf{M}$ into $S$ and $\bigO(n^2)$ inserts into $U_1, \ldots, U_m$ (at most $n^2$ per relation, and $m$ is a constant). In each of the $n$ steps, it inserts at most $n$ tuples into $T$ and exactly $n$ tuples into $R$, and it makes $2n$ enumeration requests, each returning at most one tuple. In total, the reduction executes $N = \bigO(n^2)$ inserts and $2n^2$ enumeration requests.

    Assume there exists an algorithm that maintains $Q$ with $\bigO(N^{\frac{1}{2}-\gamma})$ amortized update time and $\bigO(N^{\frac{1}{2}-\gamma})$ enumeration delay, for some $\gamma>0$. Then the $N$ inserts take at most $N \cdot \bigO(N^{\frac{1}{2}-\gamma}) = \bigO(n^{3-2\gamma})$ time in total, and the $2n^2$ enumeration requests take $2n^2 \cdot \bigO(n^{1-2\gamma}) = \bigO(n^{3-2\gamma})$ time in total. Together with the $\bigO(n^2)$ time for the additions and equality tests that recover $\mathbf{M}\mathbf{v}_k$ from the enumerated values, the OMv problem is solved in time $\bigO(n^{3-2\gamma}) + \bigO(n^{3-2\gamma}) + \bigO(n^2) = \max(\bigO(n^{3-2\gamma}), \bigO(n^2))$, which is subcubic for every $\gamma>0$ and contradicts the OMv conjecture.

    \subsubsection{Case 2} Now, assume that the query $Q$ violates the bound-free interaction of Def.~\ref{def:p-hierarchical}. This means $Q$ is of the form $Q(X,\vect F) = \sum_{\vect B}S(X,Y,\vect Z_s)\cdot T(Y,\vect Z_T)\cdot U_1(\vect Z_1)\cdot\ldots\cdot U_m(\vect Z_m)$, where $\vect F$ denotes the free variables other than $X$, $Y \notin \vect F$, and $X \notin \vect Z_T$. Indeed, the violation gives a bound variable $Y$ and a free variable $X$ with $\at(Y) \not\subseteq \at(X)$ and $\at(Y) \cap \at(X) \neq \emptyset$: let $S$ denote an atom in $\at(Y) \cap \at(X)$, $T$ an atom in $\at(Y) \setminus \at(X)$, and $U_1, \ldots, U_m$ the remaining atoms. Again, we reduce the OMv problem to the incremental maintenance of $Q$. 
    
    \paragraph*{Database Construction.} Let $\mathbf{M}$ be the $n \times n$ Boolean matrix in the OMv problem. We encode $\mathbf{M}$ into the $\mathbb{N}$-relation $S$ using at most $n^2$ inserts: for every $\mathbf{M}_{ij}=1$, we insert $S(i,j,\star)\mapsto 1$. For relations $U_1(\vect Z_1),\dots,U_m(\vect Z_m)$, we apply the general auxiliary setup from Case 1. Thus, for any $i\in[n]$ we have that $Q(i,\boldsymbol{\star})$ is equivalent to its subquery $ Q''(i) = \sum_{j\in[n]} S(i,j,\boldsymbol{\star}) \cdot T(j,\boldsymbol{\star})$ on the above-constructed database.

    At each step $k \in [n]$, vector $\mathbf{v}_k$ arrives. For every $i \in [n]$ such that $\mathbf{v}_k[i]=1$, we insert $T(i,\boldsymbol{\star})\mapsto 1$. As before, the payload of a tuple $(i,\boldsymbol{\star})$ in $T$ at step $k$ is exactly the sum of its insertions across all steps $1$ to $k$. Let this cumulative vector be $\vect t_{k} = \sum_{j\in[k]} \mathbf{v}_j$. Thus, $T(i,\boldsymbol{\star}) = \vect t_{k}[i]$. Again, we treat vectors $\vect v_1,\dots.\vect v_n$ as integer vectors.

    \paragraph*{Recovering the Result for the OMv Problem from the Query Result.}
    We enumerate the result of the query $Q(X,\boldsymbol{\star})$. For a specific output binding $X = i$, the query evaluates to:
    \[
    Q(i,\boldsymbol{\star}) = \sum_{j\in[n]} S(i,j,\boldsymbol{\star}) \cdot T(j,\boldsymbol{\star}) = \sum_{j\in[n]} \mathbf{M}_{ij} \vect t_{k}[j].
    \]
    Thus, enumerating all $n$ results of the query yields the complete integer vector $\mathbf{M} \vect t_{k}$. 

    Since $\vect t_k - \vect t_{k-1} = \vect v_k$ (with $\vect t_0 = \vect 0$), the $i$-th entry of the required Boolean OMv answer $\vect M \vect v_k$ is $(\vect M \vect t_k)[i] - (\vect M \vect t_{k-1})[i]$. This entry is nonnegative, so it is nonzero iff $(\vect M \vect t_k)[i] \neq (\vect M \vect t_{k-1})[i]$. The reduction stores the vector enumerated at the previous step and outputs for each $i \in [n]$ the Boolean result of this equality test. The recovery thus uses only equality tests on enumerated answers.

    \paragraph*{Time Analysis.} Starting from the empty database, the reduction executes at most $n^2$ inserts to encode $\mathbf{M}$ into $S$ and $\bigO(n^2)$ inserts into $U_1, \ldots, U_m$ (at most $n^2$ per relation, and $m$ is a constant). In each of the $n$ steps, it inserts at most $n$ tuples into $T$ and makes one enumeration request. In total, the reduction executes $N = \bigO(n^2)$ inserts and $n$ enumeration requests.

    Assume there exists an algorithm that maintains $Q$ with $\bigO(N^{\frac{1}{2}-\gamma})$ amortized update time and $\bigO(N^{\frac{1}{2}-\gamma})$ enumeration delay, for some $\gamma>0$. Then the $N$ inserts take at most $N \cdot \bigO(N^{\frac{1}{2}-\gamma}) = \bigO(n^{3-2\gamma})$ time in total, and the $n$ enumeration requests, each of which returns at most $n$ tuples, take $n \cdot (n+1) \cdot \bigO(n^{1-2\gamma}) = \bigO(n^{3-2\gamma})$ time in total. Together with the $\bigO(n^2)$ time for the equality tests that recover $\mathbf{M}\mathbf{v}_k$ from the enumerated vectors, the OMv problem is solved in time $\bigO(n^{3-2\gamma}) + \bigO(n^{3-2\gamma}) + \bigO(n^2) = \max(\bigO(n^{3-2\gamma}), \bigO(n^2))$, which is subcubic for every $\gamma>0$ and contradicts the OMv conjecture.

    This concludes the proof for the natural semiring. We next show how the reduction works by means of a concrete example.

\begin{example}
    Consider the following Boolean matrix $\mathbf{M}$ and sequence of input vectors $\mathbf{v}_1$ and $\mathbf{v}_2$:

\begin{align*} 
\mathbf{M} = \begin{pmatrix} 1 & 1 & 0 \\ 0 & 1 & 1 \\ 1 & 0 & 0 \end{pmatrix} , \quad
\mathbf{v}_1 = \begin{pmatrix} 1 \\ 0 \\ 0 \end{pmatrix}, \quad \mathbf{v}_2 = \begin{pmatrix} 0 \\ 1 \\ 0 \end{pmatrix} 
\end{align*}

Our goal is to compute $\mathbf{M}\mathbf{v}_1$ and then compute $\mathbf{M}\mathbf{v}_2$, relying only on the scalar output of the 
hard query $Q()=\sum_{X,Y}R(X)\cdot S(X,Y)\cdot T(Y)$. 

Starting from the empty database, we insert into $S$ the tuples that represent the matrix $\mathbf{M}$:

\[
\begin{array}{c@{\quad}c@{\quad}l}
X & Y & \mapsto \mathsf{S}[X,Y] \\
\hline
1 & 1 & \mapsto 1 \\
1 & 2 & \mapsto 1 \\
2 & 2 & \mapsto 1 \\
2 & 3 & \mapsto 1 \\
3 & 1 & \mapsto 1 
\end{array}
\]

Initially, the relations $R$ and $T$ are empty (i.e., mapping all tuples to $0$).

\paragraph*{Step 1:} We start with the insert $T(1)\mapsto1$. Thus, we have the cumulative vector $\mathbf{t}_1 = \mathbf{v}_1 = (1,0,0)^\top$. The relation $T$ now contains a single tuple:

\[
\begin{array}{c@{\quad}l}
Y & \mapsto \mathsf{T}[Y] \\
\hline
1 & \mapsto 1
\end{array}
\]

Next, we use updates to $R$ and query evaluations to find $\mathbf{u}_1$. Initially, $Q^{(1)}_{\text{before}} = 0$.
\begin{itemize}
    \item \textbf{For $i=1$:} We insert $R(1)\mapsto 1$. $Q_{\text{before}}^{(1)} = 0, Q_{\text{after}}^{(1)} = 1 \implies Q_{\text{after}}^{(1)} - Q_{\text{before}}^{(1)} = 1 - 0 = 1$
    \item \textbf{For $i=2$:} We insert $R(2)\mapsto 1$. $Q_{\text{before}}^{(2)} = 1, Q_{\text{after}}^{(2)} = 1 \implies Q_{\text{after}}^{(2)} - Q_{\text{before}}^{(2)}  = 1 - 1 = 0$
    \item \textbf{For $i=3$:} We insert $R(3)\mapsto 1$. $Q_{\text{before}}^{(3)} = 1, Q_{\text{after}}^{(3)} = 2 \implies Q_{\text{after}}^{(3)} - Q_{\text{before}}^{(3)} = 2 - 1 = 1$
\end{itemize}

At the end of Step 1, we have that $\vect u_1 = (1,0,1)^\top$:
$$ \mathbf{u}_1 - \mathbf{u}_0 = \begin{pmatrix} 1 \\ 0 \\ 1 \end{pmatrix} - \begin{pmatrix} 0 \\ 0 \\ 0 \end{pmatrix} = \begin{pmatrix} 1 \\ 0 \\ 1 \end{pmatrix}  = \mathbf{M}\mathbf{v}_1.$$

\paragraph*{Step 2:} We process the insert $T(2)\mapsto 1$. This gives us $\mathbf{t}_2 = \mathbf{v}_1 + \mathbf{v}_2 = (1,1,0)^\top$. The relation $T$ now maps two tuples to $1$:

\[
\begin{array}{c@{\quad}l}
Y & \mapsto \mathsf{T}[Y] \\
\hline
1 & \mapsto 1 \\
2 & \mapsto 1
\end{array}
\]

Before any new updates to $R$, the query enumerates the current database state based on the accumulated payloads from Step 1. We find $Q^{(1)}_{\text{before}} = 4$. We again probe by updating $R$:
\begin{itemize}
    \item \textbf{For $i=1$:} We insert $R(1\mapsto 1)$. $Q_{\text{before}}^{1} = 4, Q_{\text{after}}^{1} = 6 \implies Q_{\text{after}}^{(1)} - Q_{\text{before}}^{(1)} = 6 - 4 = 2$
    \item \textbf{For $i=2$:} We insert $R(2\mapsto 1)$. $Q_{\text{before}}^{2} = 6, Q_{\text{after}}^{2} = 7 \implies Q_{\text{after}}^{(2)} - Q_{\text{before}}^{(2)} = 7 - 6 = 1$
    \item \textbf{For $i=3$:} We insert $R(3\mapsto 1)$. $Q_{\text{before}}^{3} = 7, Q_{\text{after}}^{3} = 8 \implies Q_{\text{after}}^{(3)} - Q_{\text{before}}^{(3)} = 8 - 7 = 1$
\end{itemize}

At the end of Step 2, we have that $\vect u_2 = (2,1,1)^\top$:
$$ \mathbf{u}_2 - \mathbf{u}_1 = \begin{pmatrix} 2 \\ 1 \\ 1 \end{pmatrix} - \begin{pmatrix} 1 \\ 0 \\ 1 \end{pmatrix} = \begin{pmatrix} 1 \\ 1 \\ 0 \end{pmatrix} = \mathbf{M}\mathbf{v}_2.$$

Because we operate over the natural semiring $(\mathbb{N}, +, \times)$, payloads accumulate monotonically. By the end of Step 2, the state of relation $R$ reflects a payload of $2$ for every probed tuple:

\[
\begin{array}{c@{\quad}l}
X & \mapsto \mathsf{R}[X] \\
\hline
1 & \mapsto 2 \\
2 & \mapsto 2 \\
3 & \mapsto 2
\end{array}
\]
\end{example}

\paragraph*{Generalization to Semirings with an Injective Homomorphism from the Natural Semiring.}
Let $h$ be an injective homomorphism from the natural semiring to a semiring $K$ (Def.~\ref{def:embedding}). We run the reduction of this section over a $K$-database, so all payloads and query answers are elements of $K$. Two parts of the reduction have to be adapted: (1) the inserts and (2) the recovery of the OMv answer from the enumerated query answers.

For (1), every insert of the original reduction has payload $1$; it becomes an insert of payload $h(1) = 1_K$. A straightforward induction shows: whenever the original reduction over $\N$ holds a value $m$, as the payload of a tuple or as an enumerated query answer, the reduction over $K$ holds $h(m)$ in its place, since $h$ preserves sums and products.

For (2), the recovery performs only additions and equality tests on the enumerated answers (Cases 1 and 2 above). The reduction over $K$ runs the same recovery on the enumerated $K$-answers, which by (1) are the values $h(a)$ for the corresponding answers $a$ over $\N$. Since $h$ preserves sums, $h(a) +_K h(b) = h(a + b)$, and since $h$ is injective, $h(x) = h(y)$ iff $x = y$; every equality test therefore has the same outcome as over $\N$, and the reduction over $K$ outputs the same bits. Additions and equality tests on $K$ take constant time (Sec.~\ref{sec:prelims}), so the time analysis is unchanged.

\subsection{The Tropical Semiring}
\label{sec:lowerbound-tropical}

We show the lower bound for the min tropical semiring $\mathbb{T} = (\mathbb{R}\cup\{+\infty\}, \min,+)$, a similar proof works for the max tropical (or arctic) semiring $(\mathbb{R}\cup\{-\infty\}, \max,+)$. The proof is different from the one for the natural semiring in Section~\ref{sec:lowerbound-natural}.

Let $Q$ be any conjunctive query without self-joins that is not p-hierarchical. 
There are two ways in which $Q$ can violate the p-hierarchical property. 
    
\subsubsection{Case 1}\label{sec:lowerbound-tropical-case1} We assume that $Q$ violates the bound-bound interaction, so it is of the form $Q(\vect F) = \min_{\vect B}R(X,\vect Z_R) + S(X,Y,\vect Z_s) + T(Y,\vect Z_T) + U_1(\vect Z_1) + \ldots + U_m(\vect Z_m)$, subject to the constraints $X,Y \notin \vect F$, $X \notin \vect Z_T$, and $Y\notin \vect Z_R$. The atoms $R$, $S$, $T$, and $U_1, \ldots, U_m$ are chosen from the violation as in Appendix~\ref{sec:lowerbound-natural}. We reduce the OuMv problem to the incremental maintenance problem for $Q$.
    
\paragraph*{Database Construction.} Similar to the previous reductions, we simplify the query $Q$ to its hard subquery $Q'() = \min_{X,Y} R(X) + S(X,Y) + T(Y)$ by fixing the variables in $\vars(Q)\setminus\set{X,Y}$ to a special constant $\star$ that is not in the domains of $X$ and $Y$.    
For conciseness, we overload $\star$ to also represent a tuple consisting entirely of this constant, with arity determined by the relation's schema. To neutralize the effect of the relations $U_1, \dots, U_m$ on the joins on $X$ and $Y$, we execute the following inserts into each $U_i(\vect Z_i)$:
    \begin{itemize}
        \item If $Z_i \cap \set{X,Y} = \emptyset$, we insert $U_i(\star)\mapsto 0$.
        \item If $Z_i \cap \set{X,Y} = \set{X}$ or $Z_i \cap \set{X,Y} = \set{Y}$, we insert $U_i(j,\star)\mapsto0$ for all $j \in [n]$, where $j$ is the value for $X$ or $Y$.
        \item If $Z_i \cap \set{X,Y} = \set{X,Y}$, we insert $U_i(j,k,\star)\mapsto 0$ for all $(j,k) \in [n]^2$, where $j,k$ are the values for $X,Y$.
    \end{itemize}
    Note that $0$ is the multiplicative identity of the tropical semiring. By construction, 
    \[Q(\star) = \min_{j,k\in[n]} R(j,\star) + S(j,k,\star) + T(k,\star).\] 
    
    Let $\mathbf{M}$ be the $n \times n$ Boolean matrix, and let $(\mathbf{u}_i, \mathbf{v}_i)$ be the sequence of vector pairs at step $i \in[n]$. We encode the matrix $\mathbf{M}$ into the $\mathbb{T}$-relation $S$ using at most $n^2$ inserts. For all pairs $(j,k) \in [n]^2$:
    \begin{itemize}
        \item If $\mathbf{M}_{jk}=1$, we insert $S(j,k,\star)\mapsto 0$.
        \item If $\mathbf{M}_{jk}=0$, we make no insert, so $S(j,k,\star) = \infty$.
    \end{itemize}
    The relations $R$ and $T$ start empty.

    At each step $i$, we receive vectors $\mathbf{u}_i$ and $\mathbf{v}_i$. Since we operate in an insert-only setting, we cannot empty the relations $R$ and $T$ using deletes. Instead, we insert new versions of tuples with strictly smaller payloads, ensuring that the fresh data dominates the minimization.
    
    We define the payload for step $i$ as $\Delta_i = -i$; thus $\Delta_i < \Delta_{i-1}$. We perform the following updates:
    \begin{itemize}
        \item For every $j$ such that $\mathbf{u}_i[j]=1$, insert $R(j,\star)\mapsto\Delta_i$.
        \item For every $k$ such that $\mathbf{v}_i[k]=1$, insert $T(k,\star)\mapsto\Delta_i$.
    \end{itemize}

    For indices where the Boolean vector is $0$ we perform no inserts; they retain the payload $\Delta_t$ of the most recent step $t < i$ at which they were inserted (where $\Delta_t > \Delta_i$), or the payload $\infty$ if they were never inserted.

    \paragraph*{Recovering the Result for the OuMv Problem from the Query Result.}
    Let $P(j,k) = R(j,\star) + S(j,k,\star) + T(k,\star)$. By construction, we have $Q(\star) = \min_{j,k} P(j,k)$. Consider the minimum path cost at step $i$. We distinguish two cases:
    \begin{itemize}
        \item \textbf{Case 1: $\mathbf{u}_i^\top \mathbf{M} \mathbf{v}_i = 1$.} 
        Thus, there exist indices $j, k$ such that $\mathbf{u}_i[j]=1$, $\mathbf{v}_i[k]=1$, and $\mathbf{M}_{jk}=1$. For this specific pair, the current payloads in $R$ and $T$ are $\Delta_i$, and $S(j,k,\star)=0$. Thus:
        \[ P(j,k) = \Delta_i + 0 + \Delta_i = 2\Delta_i \]
        Every payload in $R$ and $T$ is at least $\Delta_i$ and every payload in $S$ is at least $0$, so $P(j,k) \geq 2\Delta_i$ for every pair $(j,k)$; the answer is thus exactly $2\Delta_i$.
        
        \item \textbf{Case 2: $\mathbf{u}_i^\top \mathbf{M} \mathbf{v}_i = 0$.} 
        For any pair $(j,k)$, if $\mathbf{M}_{jk}=0$, then $S(j,k,\star)=\infty$ and $P(j,k)=\infty$.
        Consider a pair where $\mathbf{M}_{jk}=1$ (so $S(j,k,\star)=0$), but the vectors do not match (i.e., $\mathbf{u}_i[j]=0$ or $\mathbf{v}_i[k]=0$). This implies that at least one of the tuples in $R$ or $T$ was not updated in step $i$. That tuple is either stale, that is, last inserted at some step $t \le i-1$, or absent, that is, never inserted. A stale tuple has payload $\Delta_t \geq \Delta_{i-1}$, since payloads decrease, and an absent tuple has payload $\infty$. In both cases the payload is at least $\Delta_{i-1}$.
        Therefore, the path cost is bounded by:
        \[ P(j,k) \ge \Delta_{i-1} + 0 + \Delta_i = -(i-1) - i = 2\Delta_i + 1 > 2\Delta_i. \]
    \end{itemize}
    
    Thus, $\mathbf{u}_i^\top \mathbf{M} \mathbf{v}_i = 1 \iff Q(\star) = 2\Delta_i$.

    \paragraph*{Time Analysis.} Starting from the empty database, the reduction executes at most $n^2$ inserts to encode $\mathbf{M}$ into $S$ and $\bigO(n^2)$ inserts into $U_1, \ldots, U_m$ (at most $n^2$ per relation, and $m$ is a constant). In each of the $n$ steps, it inserts at most $n$ tuples into $R$ and at most $n$ tuples into $T$, and it makes one enumeration request, which returns at most one tuple. In total, the reduction executes $N = \bigO(n^2)$ inserts and $n$ enumeration requests.

    Assume there exists an algorithm that maintains $Q$ with $\bigO(N^{\frac{1}{2}-\gamma})$ amortized update time and $\bigO(N^{\frac{1}{2}-\gamma})$ enumeration delay, for some $\gamma>0$. Then the $N$ inserts take at most $N \cdot \bigO(N^{\frac{1}{2}-\gamma}) = \bigO(n^{3-2\gamma})$ time in total, and the $n$ enumeration requests take $n \cdot \bigO(n^{1-2\gamma}) = \bigO(n^{2-2\gamma})$ time in total. Together with the $\bigO(n)$ time for the comparisons that recover $\mathbf{u}_i^\top \mathbf{M} \mathbf{v}_i$ from the enumerated payloads, the OuMv problem is solved in time $\bigO(n^{3-2\gamma}) + \bigO(n^{2-2\gamma}) + \bigO(n) = \bigO(n^{3-2\gamma})$, which is subcubic for every $\gamma>0$ and contradicts the OuMv conjecture and thus the OMv conjecture.

    \subsubsection{Case 2} Now, assume that $Q$ violates the bound-free interaction, so it is of the form $Q(X,\vect F) = \min_{\vect B} S(X,Y,\vect Z_s) + T(Y,\vect Z_T) + U_1(\vect Z_1) +\ldots + U_m(\vect Z_m)$, where $\vect F$ denotes the free variables other than $X$, $Y \notin \vect F$, and $X \notin \vect Z_T$. The atoms $S$, $T$, and $U_1, \ldots, U_m$ are chosen from the violation as in Appendix~\ref{sec:lowerbound-natural}. We reduce the OMv problem to the IVM problem for $Q$. 
    
    \paragraph*{Database Construction.} As before, we use the special constant $\star$ and encode the matrix $\mathbf{M}$ into the $\mathbb{T}$-relation $S$ using at most $n^2$ inserts: for every $\mathbf{M}_{jk}=1$, we insert $S(j,k,\star)\mapsto 0$, and for every $\mathbf{M}_{jk}=0$ we make no insert, so $S(j,k,\star) = \infty$. The relation $T$ starts empty. For the relations $U_1(\vect Z_1),\dots,U_m(\vect Z_m)$, we apply the setup from Case 1 (Appendix~\ref{sec:lowerbound-tropical-case1}). Thus, 
    \[ Q(X, \star) = \min_{Y} ~ S(X, Y, \star) + T(Y, \star). \]

    At each step $i \in [n]$, we receive vector $\mathbf{v}_i$. We define the payload for step $i$ as $\Delta_i = -i$, ensuring $\Delta_i < \Delta_{i-1}$. For every index $k$, for which $\mathbf{v}_i[k]=1$, we insert $T(k,\star)\mapsto\Delta_i$. For every index $k$, for which $\mathbf{v}_i[k]=0$, we make no insert; the payload of $T(k,\star)$ is then $\Delta_t$ for the most recent step $t < i$ with $\mathbf{v}_t[k]=1$, or $\infty$ if there is no such step. In both cases it is at least $\Delta_{i-1}$. 
        
    \paragraph*{Recovering the Result for the OMv Problem from the Query Result.}
    Consider the value of the join for a specific $j$:
    \[ P(j) = \min_{k} \left( S(j,k) + T(k) \right). \]
    \begin{itemize}
        \item \textbf{Case 1: $(\mathbf{M}\mathbf{v}_i)_j = 1$.} 
        There exists some $k$ such that $\mathbf{M}_{jk}=1$ and $\mathbf{v}_i[k]=1$. For this $k$, $S(j,k)=0$ and $T(k)$ has the fresh payload $\Delta_i$. 
        \[ P(j) \le 0 + \Delta_i = \Delta_i. \]
        Every payload in $T$ is at least $\Delta_i$ and every payload in $S$ is at least $0$, so also $P(j) \geq \Delta_i$; hence $P(j) = \Delta_i$.
        
        \item \textbf{Case 2: $(\mathbf{M}\mathbf{v}_i)_j = 0$.}
        For all $k$, either $\mathbf{M}_{jk}=0$ (so $S(j,k)=\infty$) or $\mathbf{v}_i[k]=0$ (so $T(k)$ is stale or absent, with payload $\ge \Delta_{i-1}$).
        \[ P(j) \ge 0 + \Delta_{i-1} > \Delta_i. \]
    \end{itemize}
    By filtering for tuples where the payload is exactly $\Delta_i$, we correctly recover exactly the indices $j$ where the result vector is $1$.

    \paragraph*{Time Analysis.} Starting from the empty database, the reduction executes at most $n^2$ inserts to encode $\mathbf{M}$ into $S$ and $\bigO(n^2)$ inserts into $U_1, \ldots, U_m$ (at most $n^2$ per relation, and $m$ is a constant). In each of the $n$ steps, it inserts at most $n$ tuples into $T$ and makes one enumeration request. In total, the reduction executes $N = \bigO(n^2)$ inserts and $n$ enumeration requests.

    Assume there exists an algorithm that maintains $Q$ with $\bigO(N^{\frac{1}{2}-\gamma})$ amortized update time and $\bigO(N^{\frac{1}{2}-\gamma})$ enumeration delay, for some $\gamma>0$. Then the $N$ inserts take at most $N \cdot \bigO(N^{\frac{1}{2}-\gamma}) = \bigO(n^{3-2\gamma})$ time in total, and the $n$ enumeration requests, each of which returns at most $n$ tuples, take $n \cdot (n+1) \cdot \bigO(n^{1-2\gamma}) = \bigO(n^{3-2\gamma})$ time in total. Together with the $\bigO(n^2)$ time for the comparisons that recover $\mathbf{M}\mathbf{v}_i$ from the enumerated payloads, the OMv problem is solved in time $\bigO(n^{3-2\gamma}) + \bigO(n^{3-2\gamma}) + \bigO(n^2) = \bigO(n^{3-2\gamma})$, which is subcubic for every $\gamma>0$ and contradicts the OMv conjecture.

\paragraph*{Generalization to Idempotent Strictly Ordered Semirings.}
Let $K = (K, +_K, \cdot_K, 0_K, 1_K)$ be an idempotent strictly ordered semiring with chain $0_K = c_0 \prec_K c_1 \prec_K c_2 \prec_K \cdots$ (Def.~\ref{def:temporal-chain-semiring}). The tropical proofs of this section reason about payloads with the arithmetic of numbers. In $K$, this role is taken by two properties of the natural order (Sec.~\ref{sec:class-k}), each immediate from the definition $a \preceq_K b$ iff $a + b = b$:
\begin{enumerate}[(i)]
    \item $a + b$ is the least upper bound of $a$ and $b$, since $a + (a + b) = a + b$, and $a + b + c = c$ whenever $a + c = c$ and $b + c = c$;
    \item multiplication is monotone, since $a + b = b$ implies $a \cdot c + b \cdot c = b \cdot c$.
\end{enumerate}

We translate the payloads of the tropical reductions of both cases, which are elements of $\mathbb{T}$, into elements of $K$. The reductions use three kinds of payloads: (1) absent tuples have payload $\infty$, the additive identity of $\mathbb{T}$; in $K$, they have payload $0_K$; (2) the tuples of $S$ and of $U_1, \ldots, U_m$ have payload $0$, the multiplicative identity of $\mathbb{T}$; in $K$, they get payload $1_K$; (3) the tuples inserted at step $i$ have payload $\Delta_i$; in $K$, they get payload $c_i$. The database constructions and the steps of the reductions are otherwise unchanged.

It remains to redo the correctness arguments in $K$. Consider step $i$ of the reduction of Case 1. Every payload in $R$ and $T$ is some $c_t$ with $0 \leq t \leq i$, hence at most $c_i$. The query answer is the sum of all products $R(j) \cdot S(j,k) \cdot T(k)$, that is, their least upper bound by (i).
\begin{itemize}
    \item If $\mathbf{u}_i^\top \mathbf{M} \mathbf{v}_i = 1$, there is a pair $(j,k)$ such that $R(j)$ and $T(k)$ were both inserted at step $i$ and $\mathbf{M}_{jk} = 1$; its product is $c_i \cdot c_i$. Every other product is at most $c_i \cdot c_i$ by (ii). The answer is thus $c_i \cdot c_i$.
    \item If $\mathbf{u}_i^\top \mathbf{M} \mathbf{v}_i = 0$, then in every product with $\mathbf{M}_{jk} = 1$, $R(j)$ or $T(k)$ is stale or absent, so its payload is some $c_t$ with $0 \leq t < i$, hence at most $c_{i-1}$. The product is then at most $c_{i-1} \cdot c_i$ by (ii), which is strictly below $c_i \cdot c_i$ by Def.~\ref{def:temporal-chain-semiring}. The answer is thus strictly below $c_i \cdot c_i$.
\end{itemize}
Testing whether the answer equals $c_i \cdot c_i$ therefore recovers $\mathbf{u}_i^\top \mathbf{M} \mathbf{v}_i$. In Case 2, the same argument with a single factor shows that the answer for $j$ is $c_i$ if $(\mathbf{M}\mathbf{v}_i)_j = 1$ and at most $c_{i-1} \prec_K c_i$ otherwise.

The reduction for $K$ executes the same inserts and enumeration requests as the tropical reduction. Its own work consists of computing $c_i$ from $c_{i-1}$ and $c_i \cdot c_i$ once per step, and of comparing each enumerated answer with $c_i \cdot c_i$ (Case 1) or with $c_i$ (Case 2); each of these operations takes constant time by the assumptions of Secs.~\ref{sec:prelims} and~\ref{sec:class-k}, and the number of comparisons is the same as in the tropical analyses. The time analyses and the contradictions with the OMv conjecture are therefore unchanged. For $K = \mathbb{T}$ and the chain of Example~\ref{ex:iso}, the reduction for $K$ is the tropical reduction with $\Delta_i = -i$.


\end{document}